\documentclass[%
 reprint,
 amsmath,amssymb,
 aps,
 pra,
]{revtex4-1}   

\usepackage{geometry}
\usepackage{amsmath, amsthm, amssymb, bm}
\usepackage{mathrsfs}
\usepackage{extarrows}
\usepackage{caption,subcaption}
\usepackage{xcolor}
\usepackage{tabularx}
\usepackage{adjustbox}
\usepackage{multirow}
\usepackage{graphicx}

\usepackage[font=small,labelfont=bf]{caption}
\usepackage{braket}
\usepackage{pgfplots}

\usepackage[noline,ruled,linesnumbered,noend]{algorithm2e}

\makeatletter
\renewcommand{\Indentp}[1]{%
  \advance\leftskip by #1
  \advance\skiptext by -#1
  \advance\skiprule by #1}%
\renewcommand{\Indp}{\algocf@adjustskipindent\Indentp{\algoskipindent}}
\renewcommand{\Indm}{\algocf@adjustskipindent\Indentp{-\algoskipindent}}
\makeatother

\usepackage{caption}

\usepackage{qcircuit}
\usepackage{float}
\usepackage{hyperref}
\hypersetup{  colorlinks=true, linkcolor=blue, citecolor=red, urlcolor=blue }

\DeclareMathOperator*{\argmax}{\arg\,\max}
\newcommand{\floor}[1]{\lfloor #1 \rfloor}

\begin{document}
\newcommand{\abs}[1]{\left| #1 \right|}
\newcommand{\vabs}[1]{\left\| #1 \right\|}

\title{Optimization of CNOT circuits on 
{limited connectivity architecture}
}

\author{Bujiao Wu$^{1,2,3}$}
\author{Xiaoyu He$^{1,2}$}
\author{Shuai Yang$^{1,2}$}
\author{Lifu Shou$^{1,2}$}
\author{Guojing Tian$^{1,2}$}
\author{Jialin Zhang$^{1,2}$}
\author{Xiaoming Sun$^{1,2}$}

\email{sunxiaoming@ict.ac.cn}

\affiliation{
 $^{1}$Institute of Computing Technology, CAS, Beijing, 100190, China\\
$^{2}$University of Chinese Academy of Sciences, Beijing, 100049, China\\
$^{3}$Center on Frontiers of Computing Studies, School of Computer Science, Peking University, Beijing, 100871, China
}

\begin{abstract}
{A CNOT circuit is the key gadget for entangling qubits in quantum computing systems. However, the qubit connectivity of noisy intermediate-scale quantum (NISQ) devices is constrained by their {limited connectivity architecture}. To improve the performance of CNOT circuits on NISQ  devices, we investigate the optimization of the size/depth of CNOT circuits under the {limited connectivity architecture}. We present a method that can optimize the size of any $n$-qubit CNOT circuit $O\left(\frac{n^2}{\log \delta}\right)$ on any connected graph with minimum degree $\delta$, and prove this bound is optimal for the regular graph. For the near-term sparsely connected structure, we additionally present a method that can optimize the size of any $n$-qubit CNOT circuit to below $2n^2$. The numerical experiment shows that our method performs better than state-of-the-art results. Specifically, we present an example to illustrate the applicability of our algorithm. For the grid structure, which is commonly used in current quantum devices, we demonstrate that the depth of any $n$-qubit CNOT circuit can be optimized to be linear in $n$ with certain ancillary qubits (ancillas). Experimental results indicate that this method has significant improvements compared with all of the existing methods. We additionally test our algorithms on the five-qubit IBMQ devices, and the experiments show that the measurement results of the optimized circuit with our algorithm are more robust to noise compared with the IBM mapping method.}
\end{abstract}

\maketitle

\newtheorem*{definition}{\indent Definition}
\newtheorem*{observation}{\indent Observation}
\newtheorem{lemma}{\indent Lemma}
\newtheorem{theorem}{\indent Theorem}
\newtheorem*{corollary}{\indent Corollary}
\newtheorem*{conjecture}{\indent Conjecture}
\newtheorem*{fact}{\indent Fact}
\newtheorem*{claim}{\indent Claim}

\def\QEDclosed{\mbox{\rule[0pt]{1.3ex}{1.3ex}}}
\def\QED{\QEDclosed}
\def\proof{\indent\textbf {Proof}.}
\def\endproof{\hspace*{\fill}~\QED\par\endtrivlist\unskip}

 \newcommand{\bujiao}[1]{{\color{blue} #1}}

 \newcommand{\pbra}[1]{\left( #1 \right)}
 \newcommand{\cbra}[1]{\left\{ #1 \right\}}
\newcommand{\sbra}[1]{\left[ #1 \right]}
\newcommand{\abra}[1]{\left\langle #1 \right\rangle}
\newcommand{\bmat}[1]{\begin{bmatrix} #1 \end{bmatrix}}
\newcommand{\bin}{\{0,1\}}

 \newcommand{\GL}{\mathrm{GL}}
\newcommand{\CNOT}{\mathtt{CNOT}}
 \newcommand{\Abb}{\mathbb{A}}
  \newcommand{\Cbb}{\mathbb{C}}
\newcommand{\Fbb}{\mathbb{F}}
\newcommand{\Mbb}{\mathbb{M}}
\newcommand{\Nbb}{\mathbb{N}}
\newcommand{\Rbb}{\mathbb{R}}
\newcommand{\Sbb}{\mathbb{S}}
\newcommand{\Tbb}{\mathbb{T}}
\newcommand{\Ubb}{\mathbb{U}}
\newcommand{\Vbb}{\mathbb{V}}
\newcommand{\Zbb}{\mathbb{Z}}
\newcommand{\Acal}{\mathcal{A}}
\newcommand{\Bcal}{\mathcal{B}}
\newcommand{\Ccal}{\mathcal{C}}
\newcommand{\Ecal}{\mathcal{E}}
\newcommand{\Fcal}{\mathcal{F}}
\newcommand{\Ical}{\mathcal{I}}
\newcommand{\Tcal}{\mathcal{T}}
\newcommand{\Gscr}{\mathscr{G}}
\newcommand{\Lscr}{\mathscr{L}}
\newcommand{\Amat}{\bm{\mathrm A}}
\newcommand{\Bmat}{\bm{\mathrm B}}
\newcommand{\Cmat}{\bm{\mathrm C}}
\newcommand{\Dmat}{\bm{\mathrm D}}
\newcommand{\Fmat}{\bm{\mathrm F}}
\newcommand{\Lmat}{\bm{\mathrm L}}
\newcommand{\Pmat}{\bm{\mathrm P}}
\newcommand{\Rmat}{\bm{\mathrm R}}
\newcommand{\Tmat}{\bm{\mathrm T}}
 \newcommand{\Mmat}{\bm{\mathrm M}}
\newcommand{\Imat}{\bm{\mathrm I}}
\newcommand{\Umat}{\bm{\mathrm U}}
\newcommand{\Vmat}{\bm{\mathrm V}}
\newcommand{\Wmat}{\bm{\mathrm W}}
\newcommand{\Xmat}{\bm{\mathrm X}}
\newcommand{\Ymat}{\bm{\mathrm Y}}
\newcommand{\Zmat}{\bm{\mathrm Z}}
\newcommand{\onemat}{\bm{\mathrm 1}}
\newcommand{\zeromat}{\bm{\mathrm 0}}
\newcommand{\Tsfmat}{\bm{\mathsf T}}
\newcommand{\Rsfmat}{\bm{\mathsf R}}
\newcommand{\xmat}{\bm{\mathrm X}}

\section{Introduction}
Quantum circuit synthesis is a process to construct a quantum circuit that implements a desired unitary operator and optimizes its size/depth in terms of a given gate set, which is an important task in the field of quantum computation and quantum information~\cite{ShendeBM06,vartiainen2004efficient,nielsen2002quantum}.
{There are two key limitations of the current intermediate-scale quantum devices. First, the performance and reliability of quantum devices heavily depend on the length of time that the underlying quantum states can remain coherent. Hence it is natural to design the algorithm with less coherent time and environmental noise~\cite{boixo2018characterizing, Arute2019}, in other words, with smaller size and/or lower depth. 
The second limitation is that the state-of-art quantum devices do not support placing 2-qubit gates (usually the CNOT gates) in arbitrary pairs of qubits. The 2-qubit gates are only allowed to be placed between the adjacent qubits~\cite{ibmqexp2017,ye2019propagation, Arute2019}. 
We denote a qubit as a vertex, and use
an edge between two vertices to represent a 2-qubit gate that can be performed on these two qubits. Then the limitations of the qubit connection can be represented as a \emph{{limited connectivity architecture}}.
The {limited connectivity architecture} for the near-term devices is usually selected as grid-style graphs ~\cite{ibmqexp2017,ye2019propagation,boixo2018characterizing}. 
 Meanwhile, the $d$ dimensional grid is also a good candidate for {limited connectivity architecture} for quantum supremacy by Harrow \emph{et al.}~\cite{harrow2018approximate}.}

CNOT circuits, in which there are only CNOT gates, are indispensable for quantum circuit synthesis to construct general circuits~\cite{aaronson2004improved,patel2008optimal,barends2014superconducting,vartiainen2004efficient}, since people often use CNOT gates with some single-qubit gates to build universal quantum computing~\cite{shi2002both,barenco1995elementary,boykin2000new}. The optimization of the size/depth of CNOT circuits shed some light on the more general problem of arbitrary circuit mapping.
CNOT circuits also dominate stabilizer circuits, which play an important role in quantum error correction~\cite{nielsen2002quantum} and quantum fault-tolerant computations~\cite{bravyi2005universal}. Aaronson \textit{et al.} \cite{aaronson2004improved} proved that any stabilizer circuit has a canonical form, \textit{i.e.}, H-C-P-C-P-C-H-P-C-P-C, where H and P are one layer of Hadamard gates and Phase gates respectively, and C is a block of CNOT gates only. {It follows that the depth of stabilizer circuits equals $5d_{\text{CNOT}} + 5$, where $d_{\text{CNOT}}$ is the depth of CNOT circuits.}
Hence, the optimization of CNOT circuits can be directly applied to the optimization of stabilizer circuits.
 
 Many researchers are aiming at reducing the size/depth of CNOT circuits without {limited connectivity architecture}~\cite{moore2001parallel,jiang2019optimal,patel2008optimal}.
 {For instance, Patel \emph{et al}.~\cite{patel2008optimal} proposed an algorithm to optimize any CNOT circuits to $O(n^2/\log n)$ size on the full connectivity architecture.} Moore and Nilsson \cite{moore2001parallel} proposed an algorithm to parallelize any CNOT circuits to $O(\log n)$ depth with $O(n^2)$ ancillas { on the full connectivity architecture}, in which the depth matches the lower bound $\Omega(\log n)$.

However, {these} work can not be directly applied to near-term quantum devices with {limited connectivity architecture}.

There are several size optimization algorithms for CNOT circuits under the {limited connectivity architecture}.
Kissinger \emph{et al.}~\cite{kissinger2019cnot} proposed an algorithm that gives a $2n^2$-size equivalent circuit for any CNOT circuits if the { architecture} contains a Hamiltonian path. Unfortunately, their optimized size is $O(n^{3})$ for the { architecture} without a Hamiltonian path. Nash \emph{et al.}~\cite{nash2019quantum} proposed a similar algorithm simultaneously which gives a $4n^2$-size equivalent CNOT circuit for any CNOT circuits under any connected graph.
{There arises the following question:}

{\emph{Given any CNOT circuit, how can we implement it on the state-of-art quantum devices with the smallest size of quantum gates and/or lowest circuit depth?}}

In this paper, we first consider how to optimize the size of the CNOT circuit without ancillae. Our algorithm achieves a worst-case bound $O\left(\frac{n^2}{\log \delta}\right)$-size on any connected graph with minimum degree $\delta$.
{Our algorithm is the generalization of Patel \emph{et al.}~\cite{patel2008optimal}.}
Furthermore, we prove this bound is tight for the regular graph. For the sparse graph with maximum degree $O(1)$, we slightly improve the results of Kissinger \emph{et al.}~\cite{kissinger2019cnot} and Nash \emph{et al.}~\cite{nash2019quantum}. Specifically, we propose an algorithm that can optimize the size of any given CNOT circuit to $2n^2$ on any connected graph. We simulate this size optimization algorithm on some particular graphs in near-term quantum devices, and the simulation experimental results show our optimized size is smaller than the existing results~\cite{kissinger2019cnot,nash2019quantum}. 

Secondly, based on the rapid decoherence of the quantum system and the grid constriction of the recent quantum devices~\cite{Arute2019,boixo2018characterizing}, we propose an algorithm that can optimize the depth of any given CNOT circuit to $O\pbra{\frac{n^2}{\min\cbra{m_1, m_2}}}$ with $3n \leq m_1m_2\leq n^2$ qubits on $m_1 \times m_2$ grid. The optimized depth is asymptotically tight when $m_1m_2=n^2$.
We generalize the result to any constant $d$ dimensional grid. We also give the experimental result of the depth optimization algorithm on an $n\times n$ grid. As the number of qubits grows, the optimized depth has significant advantages over the existing size optimization algorithms.
{We give the comparison of our algorithms and existing algorithms in Table \ref{tab:comp_algorithms}.
}

\begin{table*}[t]
    \centering
    \caption{Comparison of the size/depth optimization algorithms for CNOT circuit on limited connectivity architectures, where $\delta$ is the minimum degree of the architecture.}
    \begin{tabular}{c|c|c|c|c}
    \hline\hline
    Algorithms  & Architecture & Optimized size & Optimized depth & Ancillas\\
    \hline
     Patel et al. & Full connectivity  & $O(n^2/\log n)$ & $O(n^2/\log n)$ & 0\\
     \hline
     Jiang et al. & Full connectivity & $O(n^2/\log n)$ & $O(\log n +\frac{n^2}{(n+m)\log (n+m)})$
    & $m$ \\
     \hline 
     Kissinger et al. & With a Hamiltonian path & $2n^2$ & $2n^2$ & 0\\
       \hline
     Nash et al. & Any architecture & $4n^2$ & $4n^2$ & 0\\
       \hline
     Alg. 1 (SBE) & Any architecture & $O(n^2/\log \delta)$ & $O(n^2/\log \delta)$ & 0\\
     \hline
     Alg. 2 (ROWCOL) & Any architecture & $2n^2$ & $2n^2$ & 0\\
       \hline
     Alg. 3 (DepAncGrid) & Grid architecture & $O(n^2)$ & $O(n^2/\min(m_1,m_2))$ & $m_1 m_2-n$\\
     \hline\hline
    \end{tabular}
    \label{tab:comp_algorithms}
\end{table*}

{In the rest of the paper, we cover the basic notation of this paper, and the basic preliminaries of the CNOT circuit and its properties in Sec. \ref{sec:pre}. In Sec. \ref{sec:sizeOpt_gen} we introduce a size optimization algorithm, and the lower bounds on the general graph. We also give a size optimization algorithm in Sec. \ref{sec:sizeOpt_near}, together with the numerical comparison of our algorithms and existing algorithms. Additionally, we give an example to show the application of 
the algorithm.
Next, we introduce our depth optimization and the experimental results in Sec. \ref{sec:DepthOpt}. In Sec. \ref{sec:experimentIBMQ}, we implement the optimized CNOT circuit on the IBMQ device. The experimental results show fewer outcome errors compared to the original circuit on the IBMQ noisy device.
Finally, we give a discussion in Sec. \ref{sec:discuss}.
}

{\section{Preliminary}}
\label{sec:pre}

\textbf{Notations.}
We use $[n]$ to denote $\{1,2,\dots,n\}$, $\Fbb_q$ to denote field with $q$ elements, $\oplus$ to denote addition under $\Fbb_2$, $\GL(n,2)$ to denote set of all $n\times n$ invertible matrix over $\Fbb_2$,
{$\Imat$} to denote the identity matrix. Let $e_j$ be the $j$-th coordinate basis vector with all zeros but a $1$ in the $j$-th position. In the later sections, with a little abuse of symbols, we use a decimal number to represent the ceiling of this number when it needs to be an integer.

\begin{figure}
  \centering
   \includegraphics[width =0.5\textwidth]{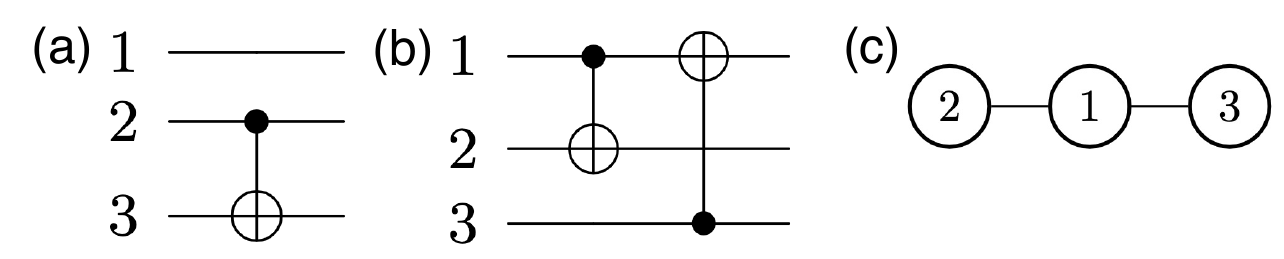}
  \caption{A 3-qubit CNOT circuit which (a) cannot be directly performed; (b) can be directly performed; {on the  limited connectivity architecture} (c).}
    \label{fig:CNOTMRep}
\end{figure}

\textbf{CNOT/SWAP Gate and Circuit.}
A CNOT gate, with control qubit $q_{i}$ and target qubit $q_{j}$, maps $\pbra{q_{i},q_{j}}$ to $\pbra{q_{i},q_{j}\oplus q_{i}}$. A SWAP gate, operating on two qubits $q_i$ and $q_j$, swaps $\pbra{q_i, q_j}$ to $\pbra{q_j, q_i}$. A CNOT circuit is a circuit that only contains CNOT gates. {We refer to a CNOT circuit with size $s$, which means the number of CNOT gates equals $s$ in this CNOT circuit.}

\textbf{Circuit with {limited connectivity architecture}.}
We use graph $G(V,E)$ to represent the {limited connectivity architecture} of two-qubit gates (CNOT gate) in the circuit. A vertex in $G$ represents a qubit, and the two-qubit gate (CNOT gate) can only operate on two qubits that are connected in $G$. 
We say a circuit $\Ccal$ is under {limited connectivity architecture} $G$ if all two-qubit gates in $\Ccal$ satisfy the {limited connectivity architecture}. Fig. \ref{fig:CNOTMRep} gives an example of the circuit {on the limited connectivity architecture}. CNOT circuit in Fig. \ref{fig:CNOTMRep} (a) cannot be performed directly on a $3$-qubit quantum device with the {limited connectivity architecture} in \ref{fig:CNOTMRep} (c). In contrast, the CNOT circuit in Fig. \ref{fig:CNOTMRep} (b) can be performed directly on it.

\textbf{Circuit with ancillas.} We say a CNOT circuit $\mathcal{C}_{n,m-n}$ with $n$-qubit inputs and $(m-n)$-qubit ancillas implements an $n$-qubit circuit $\Ccal$, if for any input $\ket{x}$ with ancillas $\ket{0}^{\otimes (m-n)}$,  the results of $\mathcal{C}_{n,m-n}$ is $\Ccal\ket{x}\otimes\ket{0}^{\otimes (m-n)}$. We say two circuits (with or without ancillas) are equivalent if they implement the same circuit $\Ccal$. 

\begin{figure}
    \centering
      \begin{tabular}[b]{c}
          \Qcircuit @C=1em @R=1.4em {
\lstick{1}&\ctrl{1}	&\targ	&\qw\\
\lstick{2}&\targ	&\qw	&\qw\\
\lstick{3}&\qw	&\ctrl{-2}	&\qw
}
\\
(a)
  \end{tabular}\qquad \quad
  \begin{tabular}[b]{c}
     $\begin{bmatrix}
      1 & 0 & 1\\
      1& 1 & 0\\
      0 & 0 & 1
      \end{bmatrix}$
      \\
      (b)
  \end{tabular}
    \caption{(a) A three-qubit CNOT circuit; (b) The matrix representation of the CNOT circuit in (a).}
    \label{fig:CNOTRep}
\end{figure}

\textbf{Matrix representation of CNOT circuit~\cite{moore2001parallel}.}
We use $\CNOT_{i,j}$ to denote CNOT gate with control qubit $q_{i}$ and target qubit $q_{j}$. The CNOT gate is an invertible linear map $\begin{bmatrix}
1&0\\
1&1\\
\end{bmatrix}$ over $\Fbb_2$.
It is easy to check, that a CNOT gate $\CNOT_{i,j}$ is equivalent to the row operation which adds the $i$-th row to $j$-th row over $\Fbb_2$ in the corresponding invertible matrix. By the linearity property of CNOT circuits, we can represent an $n$-qubit CNOT circuit $\Ccal$ as an invertible matrix $\Mmat\in\GL(n,2)$~\cite{moore2001parallel,patel2008optimal}, and the $j$-th column of $\Mmat$ equals $\Ccal e_j$.
We use $R(i,j)$ to denote such row adding operation in the matrix, and call the $\pbra{i,j}$ its index set. We take a 3-qubit CNOT circuit as an example for the matrix representation as shown in Fig. \ref{fig:CNOTRep}.
A sequence of row adding operations that transform $\Mmat$ to $\Imat$ corresponds to a CNOT circuit represented by $\Mmat^{-1}$. The size optimization of CNOT circuit $\Ccal$ is equivalent to optimizing a parameter $t$ such that there exist index pairs $(i_1,j_1),\dots,(i_t,j_t)$ satisfy $\Mmat=\prod_{k=1}^tR(i_k,j_k)$.
{Fig. \ref{fig:cnot_rowAdd} illustrates the equivalence of the CNOT gate operation on a CNOT circuit and row addition on its boolean matrix representation.}
The summation for the inputs is under module 2 in later sections. 

\begin{figure}
    \centering
    \includegraphics[trim = 0mm 0mm 0mm 35mm, clip=true, width = 0.5\textwidth]{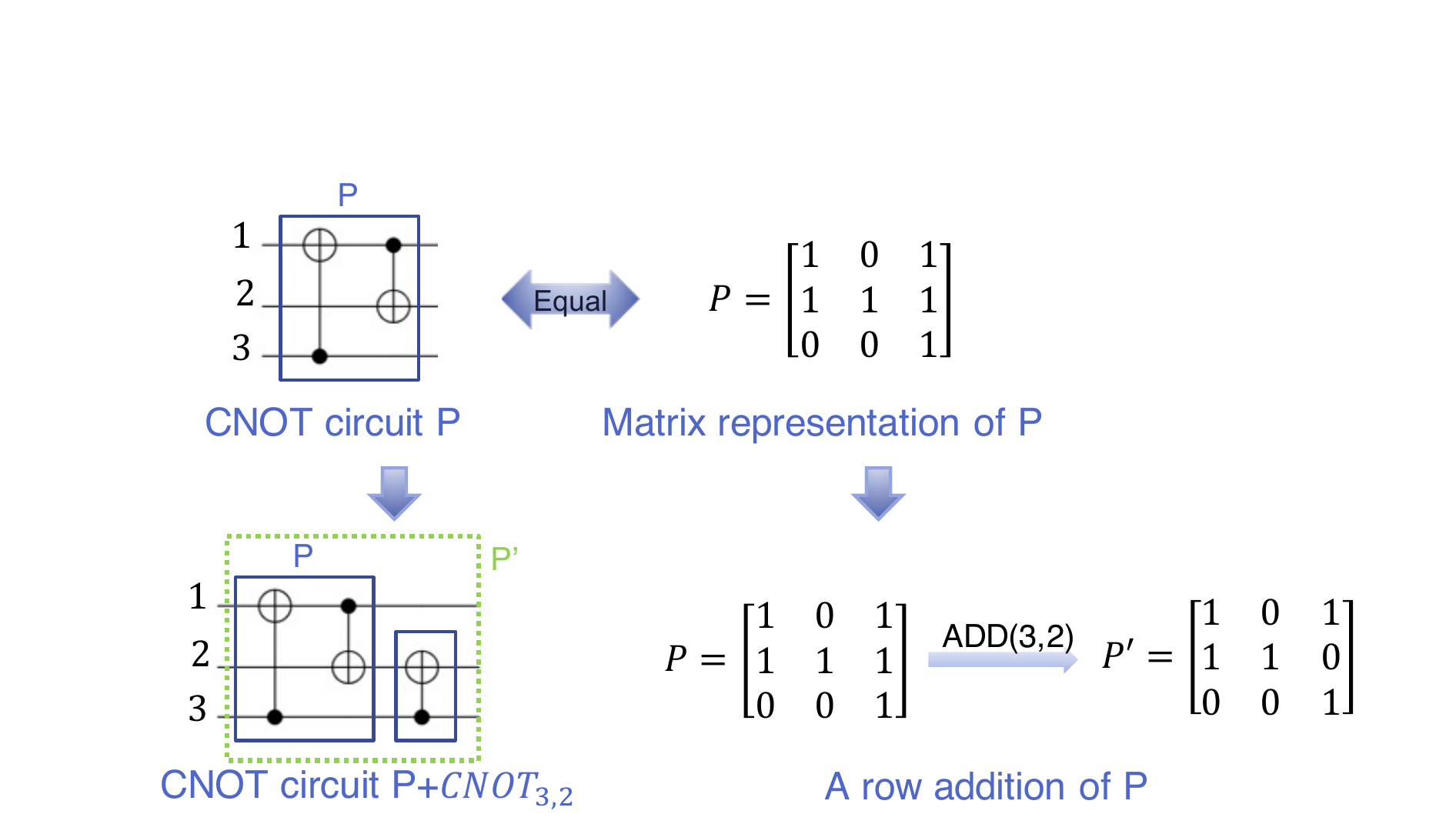}
    \caption{{Illustration of the equivalence of CNOT gate CNOT$_{3,2}$ operating on a CNOT circuit $P$ and row addition ADD$(3,2)$ on its  boolean matrix $P$.}}
    \label{fig:cnot_rowAdd}
\end{figure}

\textbf{Grid Graph.} 
In a $d$ dimensional grid graph $G(V,E)$ with
the size of each dimension be $m_i$.
A vertex in this $m_1\times \cdots \times m_d$ grid can be represented as a $d$ dimensional integer vector $(p_1,p_2,\dots,p_d)$, where $p_i\in[m_i]$. 
An $n$-qubit CNOT circuit is under $m_1\times \cdots \times m_d$ grid, which means the number of ancillas is $m_1\cdots m_d-n$, where $m_1\cdots m_d\geq n$, and the initial input $x\in\cbra{0,1}^n$ is  located on any $n$ vertices of the grid.

\textbf{Parallelized row-addition Matrices~\cite{jiang2019optimal}.}
We say a matrix $\Rmat$ is a parallelized row-addition matrix if there exists $t\in[n], \bm i\in[n]^t, \bm j\in[n]^t$ such that $i_k,j_k$'s are $2t$ different indices and $\Rmat=\prod_{k=1}^tR(i_k,j_k)$. A parallelized row-addition matrix corresponds to a CNOT circuit with depth $1$. Hence optimizing the depth of a CNOT circuit $\Ccal$ is equivalent to optimizing a parameter $t$ such that there exists $t$ parallelized row-addition matrices $\Rmat_1,\dots,\Rmat_t$ and $\Mmat=\prod_{k=1}^t\Rmat_k$.

\textbf{Several concepts in graph theory.} 
The degree of a vertex is the number of edges that are incident to the vertex. In graph $G$, $\Delta$ and $\delta$ denote the maximum and minimum degree of its vertices respectively. A graph is regular if $\Delta=\delta$. The Steiner tree with terminals set $S\subseteq V$, is a connected sub-graph of $G$ with a minimum number of edges that contain all vertices in $S$. A cut vertex is a vertex whose removal will make the connected graph disconnected.

\section{Size optimization on General graph}
\label{sec:sizeOpt_gen}

In this section, we propose an algorithm aiming at optimizing the size of CNOT circuits for quantum systems on the general {limited connectivity architecture}, as shown in Theorem \ref{thm:TopoAveDeg}. We additionally prove that our algorithm is asymptotically tight for the regular graph.

\subsection{Size optimization algorithm}

Theorem \ref{thm:TopoAveDeg} is the generalization of Patel \emph{et al.}~\cite{patel2008optimal}, which optimizes the size of CNOT circuits on the { full connectivity architecture} and gives the optimized size of $O\pbra{\frac{n^2}{\log n}}$. The most significant difference between the techniques of Theorem \ref{thm:TopoAveDeg} and the algorithms in Refs.~\cite{kissinger2019cnot,nash2019quantum} is that here we eliminate several columns simultaneously instead of eliminating a single column in each iteration.

\begin{theorem}\label{thm:TopoAveDeg}{
{Algorithm \ref{alg:nearfutureAlg} can optimize the size of any $n$-qubit CNOT circuits} to $O\pbra{\frac{n^{2}}{\log \delta}}$ under a connected graph with minimum degree $\delta$ as the {limited connectivity architecture}.}
\end{theorem}

            \begin{proof}{
    Let $k = n/\delta$ in Theorem \ref{thm:TopoAveDegNewVersion}, where $\delta$ is the minimum degree of graph $G(V,E)$. It is easy to check the summation of the degree of any $k$ vertices are greater than $n$, and thus we have an $O\pbra{\frac{n^2}{\log \delta}}$-size CNOT circuit.}
    \end{proof}
    
It follows that the optimized size in Theorem \ref{thm:TopoAveDeg} is asymptotically tight for a nearly regular graph in which all vertices have the same order of degree by Theorem \ref{thm:LowerboundSize}.
   
 \begin{widetext}
 \begin{algorithm}[!ht]
\SetKwInOut{Input}{input}\SetKwInOut{Output}{output}
\Input{{Integers $n, k$ such that $\sum_{j=1}^k d_{i_j}\geq n$ for any $k$ vertices $i_1,\ldots, i_k$, $s= \log (n/k)/2$}, matrix $\Mmat \in \Fbb_2^{n \times s}$, graph $G(V,E)$ where $|V|=n$.}
\Output{Row additions to transform $\Mmat$ into $\sbra{e_1,\ldots, e_s}$.}
    \emph{Let $\Mmat^{(1)}$ be the first $s$ rows of $\Mmat$ and $\Mmat^{(2)}$ be the rest rows, and $T$ be the 2-approximate minimum Steiner tree for vertices $\cbra{V_1,... , V_s}$ in $G$\label{step:ini_steiner}}\;
    \For{$j\leftarrow 1$ \KwTo $s$\label{step:for_j_s}}{
    \qquad   \emph{Eliminate the $j$-th column of $\Mmat^{(1)}$ to $e_j$ under graph $T$ with Lemma \ref{lem:SizeSubColSequence}}\;
    \label{step:elimi}
    }
    \emph{$l := \argmax_{s<j\leq n} \deg(V_j)$\label{step:l_argmax}}\;
    \For{$j\leftarrow 1$ \KwTo $s$\label{step:j_for_sec}}{
     \qquad\emph{Eliminate $(l,j)$-th element $\Mmat_{l,j}$ to 0 with the $j $-th row of $\Mmat^{(1)}$ under the minimum path of the vertices $V_{j}$ and $V_{l}$ with Lemma \ref{lem:SizeSubColSequence}
     \label{step:elimit_sec}}\;
    }
    \For{$i \leftarrow 1$ \KwTo $2^{s}$\label{step:for_i_exp}}{
     \qquad    \emph{Let $\text{Gray}(i) := i\oplus \floor{i/2}$ be the Gray code of $i$
     \label{step:for_grapy1}}\;
     \qquad\emph{Let $S$ be the set containing all the rows in $\Mmat^{(2)}$ with Gray code equals $\text{Gray}(i)$}\label{step:for_grapy2}\;
      \qquad   \emph{Transform the binary string of row $l$ to $Gray(i)$ by adding one of the rows in $\Mmat^{(1)}$ to row $l$ under the minimum path between the vertices of these two rows
        with Lemma \ref{lem:SizeSubColSequence}\label{step:for_grapy3}}\;
        \qquad  \While{$|S|\ne \emptyset$\label{step:while_s}}{
     \qquad \qquad   \emph{Let set $S'$ be random $k$ rows of $S$, and let set $R$ be the vertices of these rows in $S'$\label{step:while_inner1}}\;
     \qquad \qquad   \emph{Eliminate rows in $S'$ to 0 by row $l$ under the 2-approximate minimum Steiner tree of set $R$ in $G$ (Lemma \ref{lem:SizeSubColSequence})\label{step:while_inner2}}\;
    }
}
\caption{\textbf{(SBE)} {Eliminate the first $s$ columns of the given invertible matrix.}}
\label{alg:nearfutureAlg}
\end{algorithm}
\footnotetext[1]{The Gray code for row $i$ equals Gray$(i)$.}
 \end{widetext}

We give an explicit algorithm to show the upper bound in Theorem \ref{thm:TopoAveDeg}. Recall that the construction of the CNOT circuit is equivalent to constructing an invertible matrix with row addition operations by Sec. \ref{sec:pre}.
Algorithm \textbf{SBE} (size block elimination) gives the row additions for the first $s$ columns of the invertible matrix under graph $G(V,E)$. In the algorithm, we use vertex $V_i$ to represent row $i$, as well as qubit $q_i$. The $i$-th row can be directly added to the $j$-th row means vertex $V_i$ is connected to $V_j$. We also depict the process in Fig. \ref{fig:algscols}.

Let $k$ be a number such that the summation of the degree of any $k$ vertices in $G(V,E)$ are greater than the total number of qubits $n$.
Algorithm \textbf{SBE} gives an optimized {size} 
\[O\pbra{\frac{n^2}{\log (n/k)}}\leq O\pbra{\frac{n^2}{\log \delta}},\]
as stated in Theorem \ref{thm:TopoAveDegNewVersion}. Theorem \ref{thm:TopoAveDegNewVersion} is the generalization of Theorem \ref{thm:TopoAveDeg}, here we consider the first $k$ minimum degrees of the graph instead of only one minimum degree.

\begin{theorem}\label{thm:TopoAveDegNewVersion}
{Algorithm \ref{alg:nearfutureAlg} can optimize the size of any $n$-qubit CNOT circuits to $O\left(\frac{n^{2}}{\log (n/k)}\right)$ under a connected graph $G(V,E)$ as the limited connectivity architecture, where the summation of the degree of any $k$ vertices in $G$ are greater than $n$.}
\end{theorem}
    \begin{proof}
       Theorem \ref{thm:TopoAveDegNewVersion} holds by Lemma \ref{lem:MinSteinerSize} and Lemma \ref{lem:topologicalSteinerK}.
    \end{proof}

     Lemma \ref{lem:MinSteinerSize} ensures the efficiency of our optimization algorithm, by which we can eliminate one row in one step on average.

 \begin{lemma}\label{lem:MinSteinerSize}{
Given connected graph $G(V,E)$, for any integer $k$ such that the summation of the degree of any $k$ vertices is greater than $n$, the minimum Steiner tree for any $k$ vertices in $G$ is less than $5k$.}
\end{lemma}
     This Lemma can be obtained directly by applying the technique in Theorem 2.4 of Ali \emph{et al.}~\cite{ali2012upper}. The detailed proof of this lemma is in Supplementary material. The core idea of the proof is that two vertices share no common neighbors if the distance between them equals three. This lemma needs exponential cost to give a minimum Steiner tree, hence we replace it with the 2-approximate minimum Steiner tree in Algorithm \textbf{SBE}, as stated in the following corollary.
     
\begin{corollary}
{
     Given connected graph $G(V,E)$, for any integer $k$ such that the summation of the degree of any $k$ vertices are greater than $n$, the 2-approximate minimum Steiner tree for any $k$ vertices in $G$ is less than $10k$.}
\label{cor:AppSteinerSize}
\end{corollary}

The following lemma serves Lemma \ref{lem:topologicalSteinerK}. This lemma can be obtained directly from the optimization process of Nash \emph{et al.}~\cite{nash2019quantum}, by which we can add the value of a qubit to the target qubits and keep the values of the other qubits the same.

\begin{lemma}\cite{nash2019quantum}
{For $S\subseteq[n-1], y\in \{0,1\}^n$, where $y_j = x_j + x_n$ if $j\in S$, and $y_j = x_j$ otherwise. The transformation $\ket{x_1}\cdots \ket{x_n}\rightarrow \ket{y_1}\cdots \ket{y_n}$ on any connected graph can be implemented by CNOT gates with size $O\pbra{n}$. 
\label{lem:SizeSubColSequence}}
\end{lemma}
By this lemma, we can eliminate one column of the matrix with $O(n)$ row additions.

\begin{lemma}
  \label{lem:topologicalSteinerK}
 {Algorithm \ref{alg:nearfutureAlg} can optimize the size of any $n$-qubit CNOT circuits to $O\left(\frac{n^{2}}{\log (n/k)}\right)$ under a connected graph $G(V,E)$ as the limited connectivity architecture, where any $k$-vertex set generate a 2-approximate Steiner tree.}
\end{lemma}

\begin{proof}
Let $s = \log (n/k)/2$.
Given $n$-qubit CNOT circuit  $\Mmat\in\GL(n,2)$, the following algorithm uses $O\left({n^2}/{s}\right)$ row additions to transform $\Mmat$ to $\Imat$ and thus gives an equivalent $O\left({n^{2}}/{s}\right)$ size CNOT circuit.
 The algorithm starts with dividing $\Mmat$ into $n/s$ blocks $\left[\Mmat_{1}\cdots \Mmat_{n/s}\right]$, where $\Mmat_i \in \Fbb_{2}^{n\times s}$. Similarly let $\Imat = \left[\Imat_{1}\cdots \Imat_{n/s}\right]$.
Next transform $\Mmat_j$ to $\Imat_{j}$ for all of $j\in \sbra{n/s}$, as shown in Fig. \ref{fig:algscols}.

\begin{figure}[htbp]
    \centering
 \includegraphics[width =0.5\textwidth]{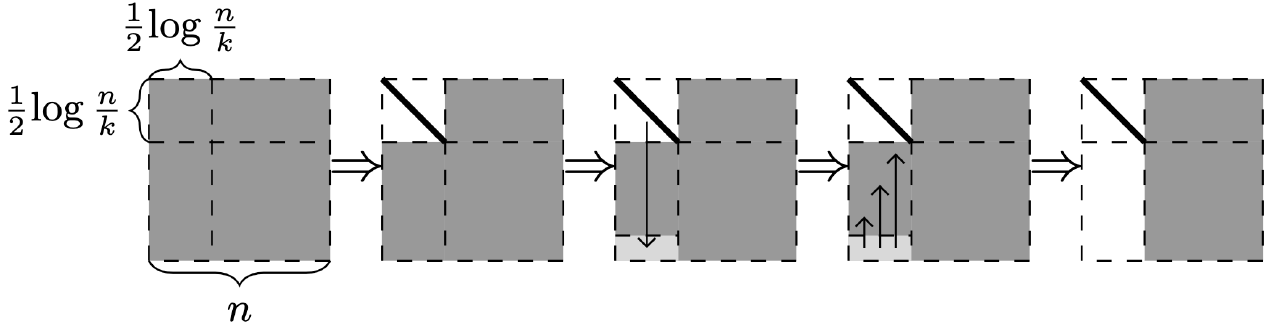} 
  	\caption{{The elimination process of Algorithm \ref{alg:nearfutureAlg}.}}
    \label{fig:algscols}
\end{figure}

Algorithm \textbf{SBE} states how to transform $\Mmat_1$ to $\Imat_1$ with row additions.
The process of transforming $\Mmat_j$ to $\Imat_j$ for $j> 1$ are almost the same with the process of transforming $\Mmat_{1}$ to $\Imat_{1}$, except in step (1) of Algorithm \textbf{SBE}, we need to change
$\Mmat^{(1)}$ into the $(j-1)s + 1$-th to $js$-th rows for input matrix $\Mmat_j$, and $\Mmat^{(2)}$ be the rest rows.

Now we prove the number of row additions is indeed $O\left(n^2/s\right)$. 
Since any $k$ vertices have an $O(k)$-size 2-approximate minimum Steiner tree, the number of vertices in 2-approximate minimum Steiner trees in Steps \ref{step:ini_steiner} and \ref{step:while_inner2} in Algorithm \textbf{SBE} are both $O(k)$. Hence the number of row additions is equal to
    \begin{align*}
        \underbrace{O(s \cdot k)}_{\text{{Steps (\ref{step:for_j_s}-\ref{step:elimi})}}} +
        \underbrace{O(s^2\cdot k)}_{\text{{Steps (\ref{step:j_for_sec}-\ref{step:elimit_sec})}}} +
        \underbrace{O(2^{s} \cdot k + n)}_{{\text{Steps (\ref{step:for_i_exp}-\ref{step:while_inner2})} }}= O(n)
    \end{align*}
  for $k\leq n$. 
  Since there are in total $n/s$ such blocks, we need
    $n/s\times O(n) = O\left({n^2}/s\right)$ row additions.
    \end{proof}

The optimized size in Theorem \ref{thm:TopoAveDegNewVersion} is asymptotically tight when $k = n^\varepsilon$ for $0\leq\varepsilon<1$, \emph{i.e.}, the summation of the degree of any $n^{\epsilon}$ vertices are greater than $n$, since $\Omega\pbra{n^2/\log n}$ is the lower bound for any {limited connectivity architecture} by Theorem \ref{thm:LowerboundSize}.
 
 {\subsection{Lower bound for size optimization}
 }
 
 The best lower bound of { full connectivity architecture} CNOT circuit synthesis is $\Omega(n^2 / \log n)$ size by Patel \emph{et al.}~\cite{patel2008optimal}. This lower bound is obtained by counting the number of distinct CNOT circuits with the given number of CNOT gates, which also implies the same lower bound $\Omega(n^2 / \log n)$ for CNOT circuit synthesis on a {limited connectivity architecture}. 
We say the quantum circuit $\hat{\Umat}\in \Cbb^{2^n\times 2^n}$ is an $\varepsilon$-approximation for the quantum circuit $\Umat\in \Cbb^{2^n\times 2^n}$ if $\vabs{\hat{\Umat} - \Umat}_2\leq \varepsilon$.
  
Here we prove a tighter lower bound  as stated in Theorem~\ref{thm:LowerboundSize}. The technique of the proof is inspired by the counting method from Christofides~\cite{christofides2014asymptotic} and Jiang \emph{et al.}~\cite{jiang2019optimal}. 

 \begin{theorem}
For any connected graph $G(V,E)$ as the {limited connectivity architecture} of the quantum system, there exist some $n$-qubit CNOT circuits for which any $\varepsilon< 1/\sqrt{2}$-approximation 2-qubit circuits need $\Omega\pbra{n^2/ \log \Delta}$ size of CNOT gates on graph $G$, where $\Delta$ is the maximum degree of the {limited connectivity architecture}.
\label{thm:LowerboundSize}
\end{theorem}

 Before giving the proof of Theorem \ref{thm:LowerboundSize}, we first define the canonical CNOT circuit.

   \begin{definition}[Canonical CNOT circuit]
    	  For any $n$-qubit CNOT circuits with $k$ CNOT gates, which can be represented as a sequence of elementary row operations, $R_1,R_2, \dots, R_k$. We say it is  canonical if and only if the sequence can be partitioned into several non-empty blocks $G_1, G_2,\dots, G_s$ and these blocks satisfy the following properties,
    	   \begin{enumerate}
    	  	\item For $1\leq i\leq s$, the index set of matrices in  block $G_i$ is disjoint with each other;
    	  	\item For $2\leq i\leq s$,  for every matrix in the block $G_i$, there exists at least one element of its index set belonging to the index set of some matrix in the previous block $G_{i-1}$.
    	  \end{enumerate}
      \label{def:CANONICAL}
    \end{definition}

 Intuitively, the canonicity  means that  CNOT gates in the same block can be executed simultaneously  and any CNOT gate in the $G_i$ block can be put into the previous block $G_{i-1}$. It is simple to prove that any CNOT circuit can be transformed into an equivalent canonical CNOT circuit  within finite steps and the readers can refer to it \cite{christofides2014asymptotic} for specific proof. Next, we will show the upper bound of the number of distinct canonical CNOT circuits as stated in Lemma \ref{lem:Canonical}.

 \begin{lemma}
  Given the  {limited connectivity architecture} $G(V,E)$, there are at most  $ 8^k e^{k} \Delta^{k} 2^{n\log n}$ different canonical  $n$-qubit CNOT circuits with $k$ CNOT gates, where  $\Delta$ is the maximum degree of the graph.
   	 \label{lem:Canonical}
 \end{lemma}  

We postpone the proof of this lemma to Supplementary material. In the following, we give the proof of Theorem \ref{thm:LowerboundSize} by Lemma \ref{lem:Canonical}.

\begin{proof}[Proof of Theorem \ref{thm:LowerboundSize}]
For a real number $a\in [-1,1]$, let the $\eta$ discretization of $a$ be $a^{\eta}=\eta \floor{\frac{a}{\eta}}$. Then there are in total $\frac{2}{\eta}$ different $\eta$ discretizations for all the continuous number $a$ in the interval $[-1,1]$ with $\abs{a^\eta - a}\leq \eta$. Hence there are in total $\pbra{\frac{2}{\eta}}^{32}$ different 
$\eta$ discretizations for all the 2-qubit unitaries $\Umat$ in $\Cbb^{4\times 4}$, and the error $\vabs{\Umat^\eta - \Umat}\leq 2\eta$.

In the following we prove that when $\eta = \frac{\varepsilon}{2k}$ and $\varepsilon<\frac{1}{\sqrt{2}}$, the $\eta$ discretizations of any two different CNOT circuits with size $k$ are different. 
Since for any two different CNOT circuits $\Umat_f, \Umat_g$ with size $k$ such that
\[
\Umat_f \ket{x} \ne \Umat_g\ket{x},
\]
we have
\[
\vabs{\Umat_f \ket{x} - \Umat_g \ket{x}}_2= \sqrt{2}.
\]
By the definition of $\Umat_f^\eta$, we have
\[\vabs{\pbra{\Umat_f^\eta - \Umat_f}\ket{x}}_2\leq 2k\eta \leq \varepsilon.
\]
Similarly we have $\vabs{\pbra{\Umat_g^\eta-\Umat_g}\ket{x}}_2\leq \varepsilon$. Since our approximation needs to be $\varepsilon<\frac{1}{\sqrt{2}}$ close to the original solution, this implies
\begin{equation}
    \Umat_f^\eta \ket{x} \ne \Umat_g^\eta \ket{x}.
    \label{eq:countingError}
\end{equation}

Hence by Lemma \ref{lem:Canonical}, we have an upper bound for the number of the 2-qubit circuit with $\varepsilon<\frac{1}{\sqrt{2}}$ approximation to all of the possible $k$ CNOT gates
\[8^k e^{k} \Delta^{k} 2^{n\log n}\pbra{\frac{2}{\eta}}^{32k}\]
with error $2k\eta$. 
Since Eq. \eqref{eq:countingError} is greater than all possible CNOT circuits, which is greater than $2^{n(n-1)/2}$, then we have
$k=\Omega\pbra{n^2/\log \Delta}$.
\end{proof}

\section{Size optimization on near-term quantum devices}
\label{sec:sizeOpt_near}
Since in the near-term quantum superconducting device, the degree of vertices on the {limited connectivity architecture} is very low, the size optimization algorithm in Sec. \ref{sec:sizeOpt_gen} gives $O\pbra{n^2}$ optimized size, which is optimal in the order. Nevertheless, the algorithm is more complex and it has a larger constant for the optimized size compared to the algorithm of Kissinger \emph{et al.}~\cite{kissinger2019cnot} and Nash \emph{et al.}~\cite{nash2019quantum}.

{The algorithms of both Kissinger \emph{et al.} and Nash \emph{et al.} aim to eliminate the $n\times n$ boolean matrix $M$ of the given CNOT circuit into identity. They both first eliminate the matrix into an upper triangular matrix by eliminating all of the ones into zeros under the diagonal line column by column, and then utilize the same method to eliminate the ones above the diagonal line. To eliminate all of the ones into zeros under the diagonal line in column $j$,
they both generate the {limited connectivity architecture} $G$ with the value of the $i$-th vertex being $M_{ij}$. Then they generate a Steiner tree with the terminals which have value one, and utilize two slightly different methods to eliminate all terminals into zero.}

In this section, we propose another size optimization algorithm that gives $2n^2$ optimized size on any connected graph in the worst case. We also show that this algorithm has a better performance compared with the existing algorithms~\cite{kissinger2019cnot,nash2019quantum}.

\subsection{Size optimization algorithm}
\label{subsec:sizeNearTerm}

  {For the ``elimination process" that transforms a matrix to identity by row operations under a {limited connectivity architecture}}, we cannot add a row to another if their corresponding vertices are not adjacent.
    Given the  {limited connectivity architecture} $G$ and the matrix $\Mmat\in \GL(n,2)$ corresponding to a CNOT circuit, in contrast to the algorithms of Kissinger \emph{et al.}~\cite{kissinger2019cnot} and Nash \emph{et al.}~\cite{nash2019quantum}, we propose an algorithm that eliminates the $i$-th row and $i$-th column simultaneously for vertex $i\in V$ which is not a cut vertex. 
The optimized size of the algorithm achieves $2n^2$ in the worst case on any connected {limited connectivity architecture}, as stated in the following theorem.

\begin{theorem}\label{thm:TopoSize2n}
{Algorithm \ref{alg:rowcol} can optimize the size of any CNOT circuits to  $2n^{2}$ under a connected graph $G(V,E)$ as the  limited connectivity architecture}.
\end{theorem}

Algorithm \textbf{ROWCOL} ensures the correctness of Theorem \ref{thm:TopoSize2n}. In Algorithm \textbf{ROWCOL}, Steps (3-6) aim to transform the $i$-th column into $e_i$ and Steps (7-10) aim to transform the $i$-th row into $e_i^T$. All arithmetic operations are over the binary field $\Fbb_2$. {An illustrative example is  shown in Example 1. The {limited connectivity architecture} for the CNOT circuit in Example 1 is depicted in Fig. \ref{fig:stairCNOT}. The optimized CNOT circuit for this inevitable matrix is the inverse of the joint circuit generated from steps (1) to (8).}

\begin{algorithm}[htbp]
\SetKwInOut{Input}{input}\SetKwInOut{Output}{output}
\Input{Integer $n$, matrix $\Mmat \in \Fbb_2^{n \times s}$, graph $G(V,E)$ where $|V|=n$.}
\Output{Row additions to transform $\Mmat$ into $\bm \Imat$.}

 \For{ $i\in V$ which is not a cut vertex}{
       \qquad \emph{$S=\{j|M_{ij\neq 0}\}\cup\{i\}$\;}
       \qquad \emph{Find a tree $T$ containing $S\subseteq V$ in graph $G$\;}
       \qquad  \emph{Postorder traverse $T$  from $i$. When reaching $j$ with parent $k$, add row $j$ to row $k$ if $\Mmat_{ji}=1$ and $\Mmat_{ki}=0$\;}    
     \qquad	\emph{Postorder traverse $T$ from $i$, add every row to its children when reached\;}
     \qquad	\emph{Let $S'\subseteq V$ that $\sum_{j\in S'}M_j=M_i+e_i$\;}
     \qquad	\emph{Find a tree $T'$ containing $S'\cup{i}$\;}
     \qquad	\emph{Preorder traverse  $T'$ from $i$. When reaching $j\notin S'$, add the $j$-th row to its parent\;}
     \qquad	\emph{Postorder traverse $T'$ from $i$, add every row to its parent when reached\;}
     \qquad	\emph{Delete $i$ from graph $G$\;}
    }
\caption{\textbf{(ROWCOL)} Near-term size optimization algorithm}
\label{alg:rowcol}
\end{algorithm}

\begin{figure*}[htbp]
\caption*{{Example 1: An illustration of Algorithm \textbf{ROWCOL}. In each block, we eliminate the red row or column on the left by leverage of the right side of the CNOT circuit.}}
\begin{tabularx}{\textwidth}{|Xm{4cm}|Xm{4cm}|}
\hline 
\small\textbf{1)} &\small\hfill &\small\textbf{2)}  & \small\hfill\\
$ \begin{pmatrix} 
\textcolor{red}{1} & 1 & 0 & 1 & 1  \\
\textcolor{red}{0} & 0 & 1 & 1 & 0  \\
\textcolor{red}{1} & 0 & 1 & 0 & 1  \\
\textcolor{red}{1} & 1 & 0 & 1 & 0  \\
\textcolor{red}{1} & 1 & 1 & 1 & 0  \\
\end{pmatrix}$
&\Qcircuit @C=1em @R=0.7em {
      \lstick{ 1} & \qw & \qw & \ctrl{4} & \qw\\
      \\
      \lstick{ 2} & \qw & \qw & \qw & \qw\\
      \lstick{ 3} & \targ & \qw & \qw & \qw\\
      \lstick{ 4} & \ctrl{-1} & \ctrl{1} & \targ & \qw\\
      \lstick{ 5} & \qw & \targ & \qw & \qw\\
      }&
$\left( \begin{array} {cccccc}
\textcolor{red}{1} & \textcolor{red}{1} & \textcolor{red}{0} & \textcolor{red}{1} & \textcolor{red}{1} \\
0 & 0 & 1 & 1 & 0 \\
0 & 1 & 1 & 1 & 1 \\
0 & 0 & 0 & 0 & 1 \\
0 & 0 & 1 & 0 & 0 
\end{array} \right)$
&\Qcircuit @C=1em @R=0.7em {
      \lstick{ 1} & \targ & \qw & \qw & \targ &\qw\\
      \lstick{ 2} & \qw & \qw & \qw & \qw &\qw\\
      \lstick{ 3} & \qw & \ctrl{1} & \qw & \qw &\qw\\
      \lstick{ 4} & \ctrl{-3} & \targ & \targ & \ctrl{-3} &\qw\\
      \lstick{ 5} & \qw & \qw & \ctrl{-1} & \qw &\qw\\
      \\
      \\
      }\\
\hline 
\small\textbf{3)} &\hfill &\small\textbf{4)}  & \small\hfill \\
$\left( \begin{array} {ccccc}
1 & 0 & 0 & 0 & 0  \\
0 & \textcolor{red}{0} & 1 & 1 & 0 \\
0 & \textcolor{red}{1} & 1 & 1 & 1  \\
0 & \textcolor{red}{1} & 0 & 1 & 0  \\
0 & \textcolor{red}{0} & 1 & 0 & 0  \\
\end{array} \right)$
&\Qcircuit @C=1em @R=0.7em {
      \lstick{ 1} & \qw & \qw & \qw & \qw\\
      \\
      \lstick{ 2} & \targ & \qw & \ctrl{1} & \qw\\
      \lstick{ 3} & \ctrl{-1} & \ctrl{1} & \targ & \qw\\
      \lstick{ 4} & \qw & \targ & \qw & \qw\\
      \\
      \lstick{ 5} & \qw & \qw & \qw & \qw\\
      }&
$\left( \begin{array} {cccccc}
1 & 0 & 0 & 0 & 0  \\
0 & \textcolor{red}{1} & \textcolor{red}{0} & \textcolor{red}{0} & \textcolor{red}{1}  \\
0 & 0 & 1 & 1 & 0  \\
0 & 0 & 1 & 0 & 1  \\
0 & 0 & 1 & 0 & 0  \\
\end{array} \right)$
&\Qcircuit @C=1em @R=0.7em {
      \lstick{ 1} & \qw & \qw & \qw & \qw\\
      \lstick{ 2} & \targ & \qw & \targ & \qw\\
      \lstick{ 3} & \ctrl{-1} & \targ & \ctrl{-1} & \qw\\
      \lstick{ 4} & \targ & \ctrl{-1} & \qw & \qw\\
      \lstick{ 5} & \ctrl{-1} & \qw & \qw & \qw\\
      \\
      \\
      }\\
\hline 
\small\textbf{5)} &\hfill &\small\textbf{6)}  & \small\hfill \\
$\left( \begin{array} {ccccc}
1 & 0 & 0 & 0 & 0  \\
0 & 1 & 0 & 0 & 0  \\
0 & 0 & \textcolor{red}{1} & 1 & 1  \\
0 & 0 & \textcolor{red}{0} & 0 & 1  \\
0 & 0 & \textcolor{red}{1} & 0 & 0  \\
\end{array} \right)$
&\Qcircuit @C=1em @R=0.7em {
      \lstick{ 1} & \qw & \qw & \qw & \qw\\
      \\
      \lstick{ 2} & \qw & \qw & \qw & \qw\\
      \\
      \lstick{ 3} & \qw & \qw & \ctrl{1} & \qw\\
      \lstick{ 4} & \targ & \ctrl{1} & \targ & \qw\\
      \lstick{ 5} & \ctrl{-1} & \targ & \qw & \qw\\
      }&
$\left( \begin{array} {cccccc}
1 & 0 & 0 & 0 & 0  \\
0 & 1 & 0 & 0 & 0  \\
0 & 0 & \textcolor{red}{1} & \textcolor{red}{1} & \textcolor{red}{1} \\
0 & 0 & 0 & 1 & 0  \\
0 & 0 & 0 & 0 & 1  \\
\end{array} \right)$
&\Qcircuit @C=1em @R=0.7em {
      \lstick{ 1} & \qw & \qw & \qw & \qw\\
      \\
      \lstick{ 2} & \qw & \qw & \qw & \qw\\
      \\
      \lstick{ 3} & \qw & \targ & \qw & \qw\\
      \lstick{ 4} & \targ & \ctrl{-1} & \qw & \qw\\
      \lstick{ 5} & \ctrl{-1} & \qw & \qw & \qw\\
      \\
      \\
      }\\
\hline 
\small\textbf{7)} &\hfill &\small\textbf{8)}  & \small\hfill \\
$\left( \begin{array} {ccccc}
1 & 0 & 0 & 0 & 0  \\
0 & 1 & 0 & 0 & 0  \\
0 & 0 & 1 & 0 & 0  \\
0 & 0 & 0 & \textcolor{red}{1} & 1 \\
0 & 0 & 0 & \textcolor{red}{0} & 1 \\
\end{array} \right)$
&\Qcircuit @C=1em @R=0.7em {
      \lstick{ 1} & \qw & \qw & \qw & \qw\\
      \\
      \lstick{ 2} & \qw & \qw & \qw & \qw\\
      \\
      \lstick{ 3} & \qw & \qw & \qw & \qw\\
      \\
      \lstick{ 4} & \qw & \qw & \qw & \qw\\
      \\
      \lstick{ 5} & \qw & \qw & \qw & \qw\\
      }&
$\left( \begin{array} {cccccc}
1 & 0 & 0 & 0 & 0  \\
0 & 1 & 0 & 0 & 0  \\
0 & 0 & 1 & 0 & 0  \\
0 & 0 & 0 & \textcolor{red}{1} & \textcolor{red}{1} \\
0 & 0 & 0 & 0 & 1 \\
\end{array} \right)$
&\Qcircuit @C=1em @R=0.7em {
      \lstick{ 1} & \qw & \qw & \qw & \qw\\\\
      \lstick{ 2} & \qw & \qw & \qw & \qw\\\\
      \lstick{ 3} & \qw & \qw & \qw & \qw\\\\
      \lstick{ 4} & \targ & \qw & \qw & \qw\\
      \lstick{ 5} & \ctrl{-1} & \qw & \qw & \qw\\
      \\
      \\
      }\\
\hline 
\small\textbf{9)} &\hfill &  & \small\hfill \\
$\left( \begin{array} {ccccc}
1 & 0 & 0 & 0 & 0  \\
0 & 1 & 0 & 0 & 0  \\
0 & 0 & 1 & 0 & 0  \\
0 & 0 & 0 & 1 & 0 \\
0 & 0 & 0 & 0 & 1 \\
\end{array} \right)$
&
&
&
\\
\hline 
\end{tabularx}
\label{tab:RCExample}
\end{figure*}

When we perform CNOT gates in Steps (2-3, 5-6), the number of CNOT gates is less than the number of remaining nodes, hence the total size is at most $4(n-1)+4(n-2)+\dots+4\times 1\le 2n^2$.
    
Since a connected graph must have a vertex that is not a cut vertex and the graph remains connected after that vertex is deleted, this algorithm can be applied to any connected graph. We can BFS (Breadth-First Search) starting from any vertex in the graph and apply the above algorithm in the BFS tree, and then we delete a leaf node each time.

  \begin{figure}[htbp]
	\centering
\includegraphics[width=0.5\textwidth]{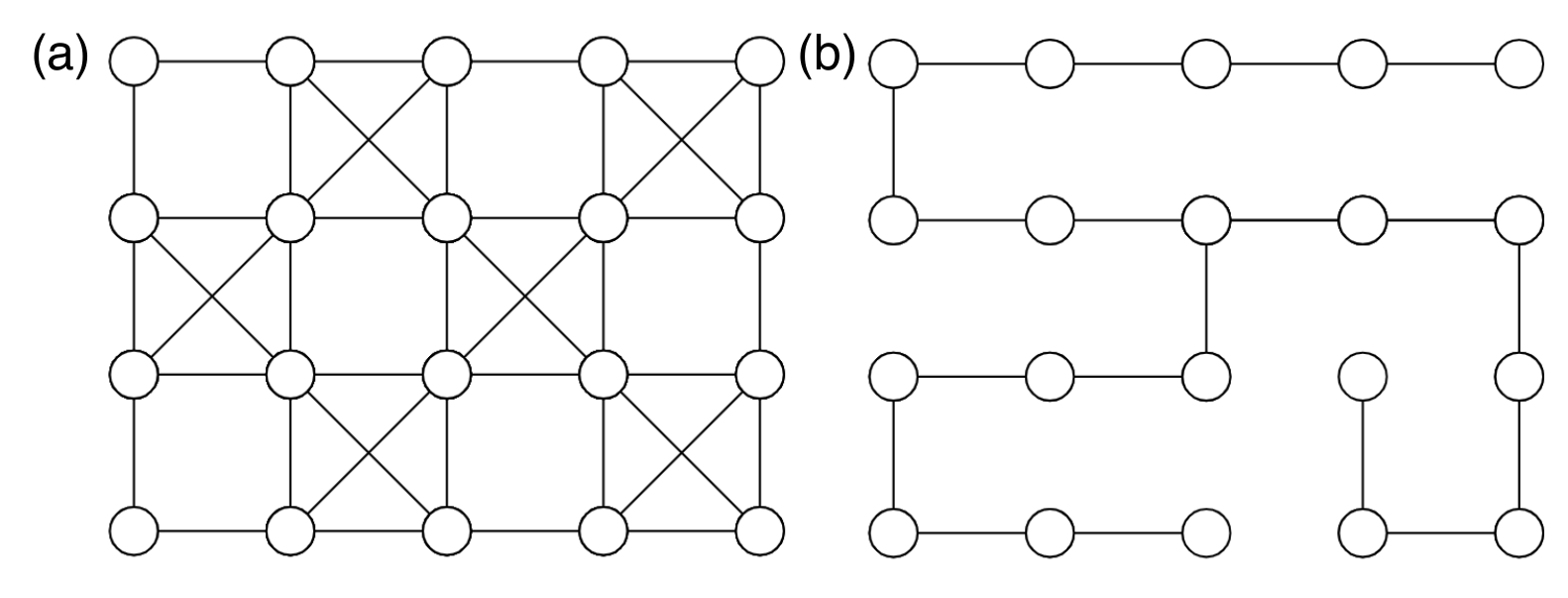}
\caption{The {limited connectivity architecture} (a) IBMQ20 and (b) T20.}
\label{fig:dif_graph}
	\end{figure}
\subsection{Numerical results} \label{subsec:SizeExperiment}

    In this subsection, we give the comparison of the experimental simulation of Algorithm \textbf{ROWCOL} and algorithms in Refs. \cite{kissinger2019cnot,nash2019quantum}.
   The CNOT circuits are performed on the IBM-Q20 graph and T-like graph (T20). The {limited connectivity architecture} of IBM-Q20 and T20 are depicted in Fig. \ref{fig:dif_graph}. 
The original circuit size ranges from 20 to 800 in the experiment, where the number of qubits is chosen to be 20.
We randomly sample 200 CNOT circuits for each input size, and compare the average optimized size for Algorithms \textbf{ROWCOL} and Ref. \cite{kissinger2019cnot,nash2019quantum}, as depicted in Fig. \ref{fig:exp_size}.

Our algorithm performs better in most generated random circuits than the algorithm of Kissinger~\emph{et al.} for the {limited connectivity architecture} with a Hamiltonian path, as shown in Fig~\ref{fig:exp_size} (a).
    In particular, we have a better-optimized size for $82.3\%$ random circuits on IBMQ20, and $99.9\%$ random circuits on the T20 graph.
    When the graph does not have a Hamiltonian path, our algorithm has more obvious advantages than Kissinger \emph{et al.}~\cite{kissinger2019cnot} (as in Fig.~\ref{fig:exp_size} (b)).
   	\begin{figure}[htbp]
        \centering
\includegraphics[width=0.5\textwidth]{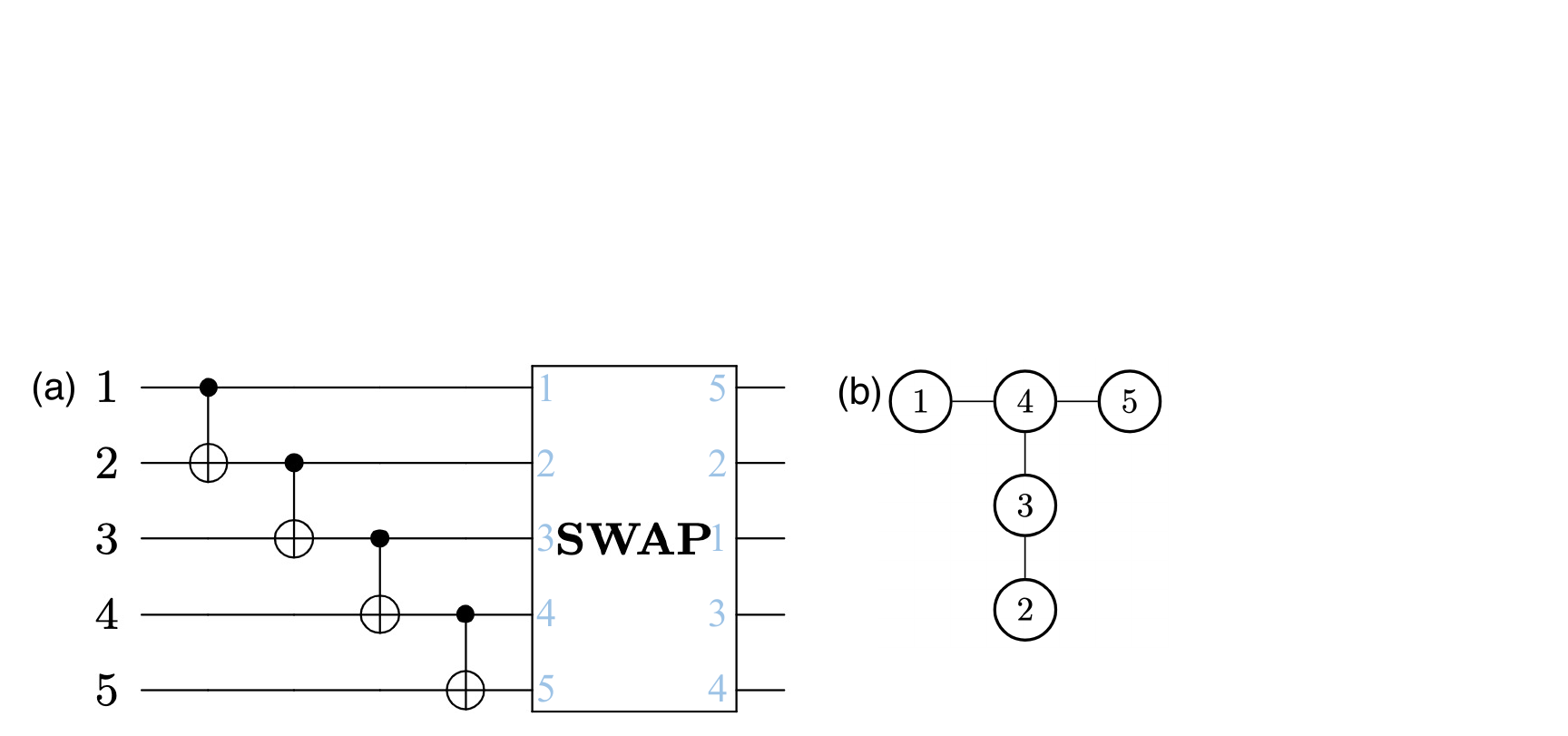}
        \caption{{(a) A staircase circuit followed by a SWAP circuit which swaps inputs $\pbra{x_1,x_2,x_3,x_4, x_5}$ into $\pbra{x_5,x_2,x_1,x_3,x_4}$;
        (b) The {limited connectivity architecture} of the quantum circuit.}}
        \label{fig:stairCNOT}
    \end{figure} 
	
	\begin{figure*}[htbp]
	\begin{center}
\includegraphics[width=0.8\textwidth]{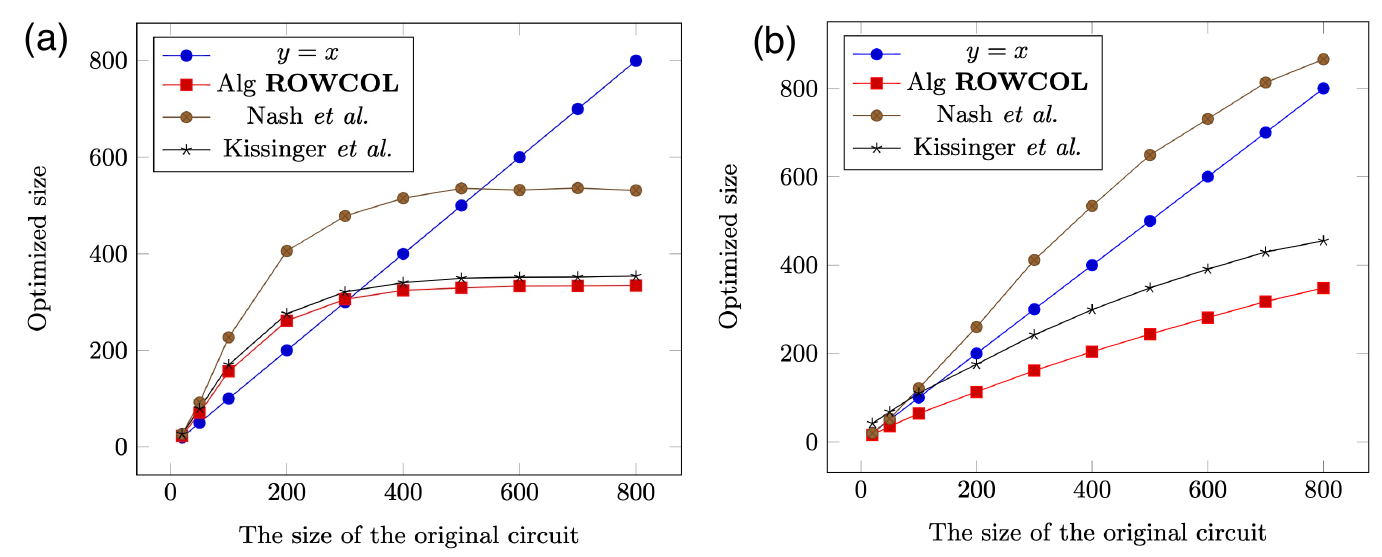}
	\caption{{The experimental results of Algorithm \textbf{ROWCOL}, algorithms in Refs. ~\cite{nash2019quantum} and ~\cite{kissinger2019cnot} under {limited connectivity architecture} graphs: (a) IBM Q20 and (b) T20.
	As a contrast, we also draw the curve $y=x$ in the figure.}}
		\label{fig:exp_size}
		\end{center}
	\end{figure*}

The above-optimized result for the randomly generated CNOT circuits shows the advantages of Algorithm \textbf{ROWCOL} for general CNOT circuits.
In the following, we perform Algorithm \textbf{ROWCOL} on a specific CNOT circuit to show the applicability of our algorithm. The comparison results are coincident with the average case.

{We also give an example to show the advantage of our algorithm, which is a staircase CNOT circuit followed by a randomly generated SWAP circuit (as shown in Fig.~\ref{fig:stairCNOT} (a)) under the {limited connectivity architecture} in Fig. \ref{fig:stairCNOT} (b). 
  Staircase CNOT and SWAP circuits both are crucial in the quantum circuit implementation of
     error detection and correction~\cite{linke2017fault,lu2008experimental,salas2004effect}, simulation of quantum chemistry~\cite{tranter2018comparison,mccaskey2019quantum}, Hamiltonian simulation~\cite{gui2020term} and near-term variational algorithms~\cite{gokhale2019partial}. 
   Fig. \ref{fig:stairCNOTAndSwap} gives a comparison of the optimized circuit between Algorithm \textbf{ROWCOL} and Algorithm in Ref. \cite{nash2019quantum}. The non-trivial optimized CNOT size of Algorithm \textbf{ROWCOL} and Ref. \cite{nash2019quantum} are $10$ and $20$ respectively. }
 
    \begin{figure*}[htbp]
        \centering
        \begin{small}
                    \begin{tabular}{c}
              \Qcircuit @C=0.6em @R=0.6em {
       \lstick{ 1} & \qw & \qw & \qw &
       \ctrl{3} & \qw & \targ & \qw & \qw & \qw & \qw\\
       \lstick{ 2} & \targ & \qw & \qw & \qw & \qw & \qw & \qw & \ctrl{1} & \qw & \qw\\
       \lstick{ 3} & \ctrl{-1} & \targ & \qw & \qw & \ctrl{1} & \qw & \ctrl{1} & \targ & \qw & \qw\\
       \lstick{4} & \qw & \ctrl{-1} & \ctrl{1} & \targ & \targ & \ctrl{-3} & \targ & \ctrl{1} & \targ & \qw\\
       \lstick{ 5} & \qw & \qw & \targ & \qw & \qw & \qw & \qw & \targ & \ctrl{-1} & \qw
       } 
     \\
            (a)
        \end{tabular}
      \quad
\begin{tabular}{c}
    \Qcircuit @C=0.6em @R = 0.6em{
    \lstick{ 1} & \qw & \qw & \qw & \targ & \qw & \qw & \qw  & \targ & \qw & \targ & \qw  & \targ & \targ & \targ & \qw & \targ & \qw & \ctrl{3} & \qw & \qw & \qw & \qw\\ 
  \lstick{ 2} & \qw & \ctrl{1} & \qw & \qw & \qw & \ctrl{1} & \qw  & \qw & \qw & \qw & \qw  & \qw& \qw & \qw & \qw  & \qw &\ctrl{1} & \qw & \qw & \qw  & \targ & \qw\\ 
  \lstick{ 3} & \ctrl{1} & \targ & \ctrl{1} & \qw & \ctrl{1} & \targ & \ctrl{1}  & \qw & \ctrl{1} & \qw & \ctrl{1}   & \qw& \qw & \qw & \qw  & \qw & \targ & \qw & \qw & \targ & \ctrl{-1} & \qw\\ 
   \lstick{ 4} & \targ & \qw & \targ & \ctrl{-3} & \targ & \qw & \targ  & \ctrl{-3} & \targ & \ctrl{-3} & \targ &\ctrl{-3} & \ctrl{-3} & \ctrl{-3} & \targ & \ctrl{-3} & \ctrl{1} & \targ & \ctrl{1} & \ctrl{-1} & \qw &  \qw\\ 
 \lstick{5}& \qw & \qw & \qw  & \qw & \qw & \qw & \qw  & \qw & \qw & \qw & \qw  & \qw & \qw & \qw  & \ctrl{-1} & \qw & \targ &\qw & \targ & \qw & \qw &  \qw
 \gategroup{1}{14}{4}{15}{.7em}{--}
    }    
           \\
            (b)
        \end{tabular}    
        \end{small}

        \caption{{Optimized circuits for the staircase CNOT circuit in Fig. \ref{fig:stairCNOT} (a) under the {limited connectivity architecture} in Fig. \ref{fig:stairCNOT} (c).
        The non-trivial optimized size of Algorithm \textbf{ROWCOL}, Ref. \cite{nash2019quantum} are 10 and 20 respectively, as shown in (a) and (b), where in (b) the two CNOT gates in the dotted box can be trivially canceled hence we did not count them in the comparison.
        }}
        \label{fig:stairCNOTAndSwap}
    \end{figure*}


\section{Depth optimization on the near-term quantum devices}
\label{sec:DepthOpt}

Due to the decoherence with the execution time of the near-term quantum devices, it is meaningful to decrease the execution time for a given quantum circuit.
Although Algorithm \textbf{ROWCOL} can also be used to optimize the depth of any given CNOT circuit, the depth of the optimized circuit is almost the same as the size.
To our knowledge, the ancillas can greatly reduce the depth of a quantum circuit. Much work has aimed to reduce the depth of CNOT circuits via designing optimized circuits with some ancillas~\cite{moore2001parallel,jiang2019optimal,brown2011ancilla}. In this section, we first propose a depth optimization algorithm on the NISQ devices --- grid with limited ancillas, and then show the great improvements of the optimized depth compared to other existing algorithms by numerical experiment.

 \subsection{Depth optimization algorithm}\label{subsec:DepthOptNear}
   Here we optimize the depth of CNOT circuits by bringing in limited ancillas on a $2$ dimensional grid. We then generalize the result to any constant $d$ dimensional grid. 
   
   \begin{theorem}\label{thm:Depth2Grid}
{
Algorithm \ref{alg:depancgrid} can optimize the depth of any CNOT circuit to 
 $O\pbra{\frac{n^2}{\min \cbra{m_1,m_2}}}$ under the $m_1\times m_2$ grid structure where $3n \leq m_1m_2\leq n^2$.}
\end{theorem}
 
 This theorem gives a trade-off of depth and ancillas for CNOT circuits under grid {limited connectivity architecture}. The depth can be optimized to $O(n)$ when $m_1=m_2=n$. It is easy to check there exists a CNOT circuit for which the optimized depth is $\Omega\pbra{n}$ on an $n\times n$ grid. Hence our optimized depth is asymptotically tight in this case.
 The main idea for Theorem \ref{thm:Depth2Grid} is to divide the output of the CNOT circuit into several intermediate results conserving the ancillas. Then combine the intermediate results to generate the output and restore the ancillas. 
 
 Before giving the algorithm to show the correctness of Theorem \ref{thm:Depth2Grid}, we would like to cast two lemmas that show how to copy and add inputs under the $d$ dimensional grid, and one can easily check that the lower bound for this operation on  an $(m_1\times m_2\times \cdots \times m_d)$ grid is also $\Omega\pbra{\sum_{j=1}^d m_j}$.

\begin{lemma}
{
Let $i\in[m]$, integer $m>1$ and $\prod_{i = 1}^d m_i = m$.
The copy operation $|0\rangle^{\otimes (i-1)}|x\rangle|0\rangle^{m-i} \rightarrow|x\rangle^m$ on the $(m_1\times m_2 \times \cdots \times m_d)$ grid with $m$ vertices can be implemented by CNOT gates with depth at most $O\pbra{\sum_{j=1}^d m_j}$, where $x\in\{0,1\}$.
\label{lem:CopyGrid}}
\end{lemma}

The following lemma gives a tight $O\pbra{dm^{1/d}}$ depth construction for addition operation on the $d$ dimensional grid. 

\begin{lemma}
{
Let $S\subseteq [m-1]$, integer $m>1, y = \sum_{i\in S} x_i \mod2$, and $\prod_{i=1}^{d} m_{i} = m$.
The addition operation $|x_1\rangle\cdots|x_{m-1}\rangle |0\rangle\rightarrow |x_1\rangle\cdots|x_{m - 1}\rangle|y\rangle$ on the above grid can be implemented by CNOT gates with depth at most $O\pbra{\sum_{i = 1}^{d}m_{i}}$, where $x_i\in\{0,1\}$ and $|y\rangle$ can be arbitrary vertices in the grid.
\label{lem:AddGrid}}
\end{lemma}
We prove this lemma by constructing a tree rooted at vertex $y$ with depth $O(\sum_{i=1}^dm_i)$ in the $d$ dimensional grid, then we divide the tree into some disjoint path to parallelize the CNOT gates.
See Supplementary material for the proofs of Lemmas \ref{lem:CopyGrid},\ref{lem:AddGrid}.

	\begin{figure*}[htbp]
\includegraphics[width=1.0\textwidth]{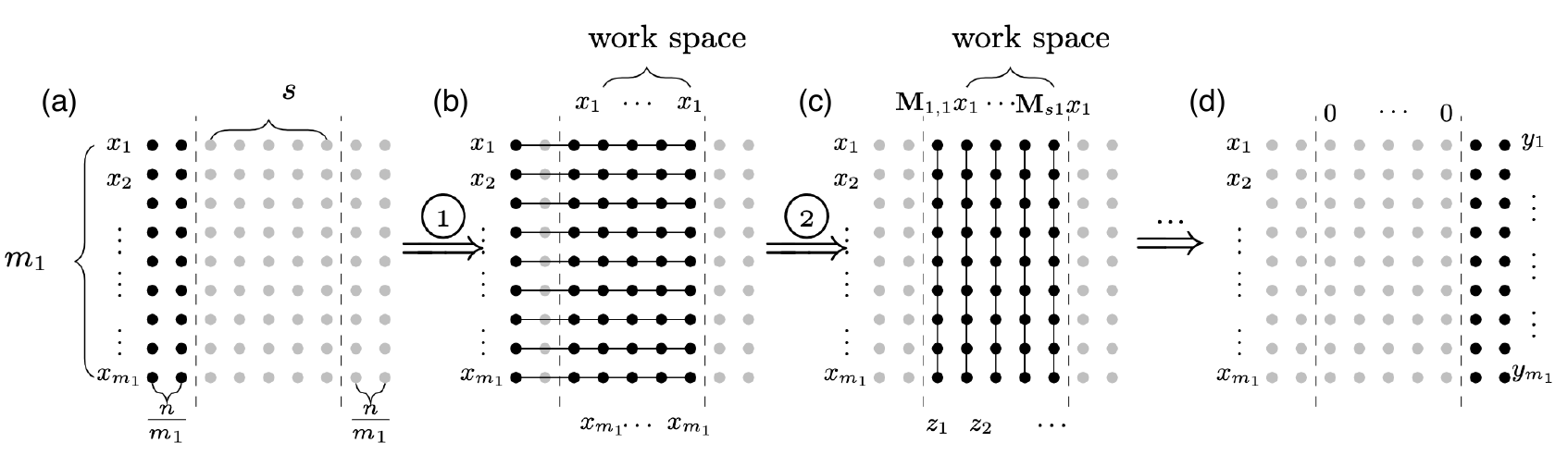}  
		\caption{{The CNOT circuitconstruction process of the algorithm with $O(\frac{n^2}{\min(m_1, m_2)})$ depth and $3n\leq m_1m_2\leq n^2$ ancillas on $m_1\times m_2$ dimensional grid, where $s:=m_2 - \frac{2n}{m_1}$.}}
		\label{fig:2dgrid}
	\end{figure*}
	
	In the following, we give the algorithm for Theorem \ref{thm:Depth2Grid}.
Let $y:= y_1\cdots y_n\in\cbra{0,1}^n$ be the output, then $y_i =\sum_{j}\Mmat_{ij}x_j$. We first divide the summation into several parts. Let $z_{ij}:= \sum_{k=(j-1)m_1 + 1}^{j m_1} \Mmat_{ik}x_k$ where $i\in[n]$ and $j \in [n/m_1]$ (Here we suppose $n/m_1$ is an integer. It is easy to generalize it to the general case.). It is easy to check that the $i$-th output qubit $y_i = \sum_{j}z_{ij}$. Let $s := m_2 - 2n/m_1$.  Let the coordinate $(i,j)$ represent the $i$-th row and $j$-th column of the $m_1\times m_2$ grid.
	
	Algorithm \textbf{DepAncGrid} implements the transformation \[\pbra{x,0^{\otimes(m-2n)},0^{\otimes n}}\xrightarrow{U_{\Mmat}} \pbra{x, 0^{\otimes (m-2n)}, \Mmat x}.\] 
    Hence, the transformation 
    \begin{equation}
        \pbra{x, 0^{\otimes (m-n)}}\rightarrow \pbra{\Mmat x, 0^{\otimes (m-n)}}
    \end{equation}
    can be implemented by first performing 
    \begin{equation}
        \pbra{x,0^{m-2n},0^{ n}}\xrightarrow{U_{\Mmat}} \pbra{x, 0^{m-2n}, \Mmat x},
    \end{equation} 
    and then performing 
    \begin{align}
          \pbra{x,0^{m-2n},\Mmat x}\xrightarrow{U_{\Mmat^{-1}}} &\pbra{x\oplus\Mmat^{-1}\Mmat x, 0^{m-2n}, \Mmat x}\\
      = &\pbra{0^{m-n}, y}.  
    \end{align}
     Finally, move $y_j$ to the first $n/m_1$ columns for all $j$ by SWAP gates.
	Hence we have an equivalent paralleled CNOT circuit for any given CNOT circuit.	We depict this process in Fig. \ref{fig:2dgrid}. 
	
	\begin{widetext}
	\begin{algorithm}[h]
	\SetKwInOut{Input}{input}\SetKwInOut{Output}{output}
	\Input{Matrix $\Mmat\in \cbra{0,1}^{n\times n}$, $m_1\times m_2$ grid, $s:=m_2 - 2n/m_1$.}
	\Output{Optimized CNOT circuit.}
	\emph{Place input $x_1,\dots, x_n$ in the first $n/m_1$ columns sequentially of the grid \tcp*{$x_1,\ldots, x_{m_1}$ in the first column, and so on, as in Fig. \ref{fig:2dgrid}}}
	\For{$l\leftarrow 1$\KwTo  $n/m_1$}{
\qquad 	\emph{Copy all of $x_i$ in the $l$-th column to the columns $j$ for $n/m_1+1\leq j \leq n/m_1 + s$ in the same row as $x_i$}\;
\qquad	\For{$j\leftarrow s$ \textbf{down} \KwTo $2$}{
\qquad	 \qquad   \If{ $\Mmat_{j,i}$ equals $0$ }{
\qquad	 \qquad     \qquad   \emph{
	        Perform CNOT$_{a,b}$, where qubit $a$ is in coordinate $(i, n/m_1 + j-1)$ and $b$ is in coordinate $(i,n/m_1 + j)$ for $i\in[m_1]$ in parallel}\;
	    }
	}
\qquad	\For{$1\leq j \leq n/s$}{
\qquad\qquad 	\For{$(j-1)s+1\leq k\leq js$ in parallel}{
\qquad\qquad\qquad	   	  \emph{Add all values of each column $k$ to the last row to generate $z_{k,1}$ in the coordinate $(m_1, n/m_1 + k)$}\;
	   }
\qquad\qquad 	\emph{Add $z_{(j-1)s+1,1},\dots, z_{js,1}$ to the mirror symmetric coordinate}\; 
	   }
\qquad\emph{Restore the ancillas in columns $j$ for $n/m_1+1\leq j \leq n/m_1 + s$}\;
	}
	\caption{\textbf{(DepAncGrid)} {Depth optimization algorithm under $m_1\times m_2$ grid}}
	\label{alg:depancgrid}
\end{algorithm}
    \end{widetext}
	The total number of paralleled operations in Algorithm \textbf{DepAncGrid} equals 
	\[
c\frac{n}{m_1}\pbra{\frac{n}{s}\pbra{m_1 + m_2}}=O\pbra{\frac{n^2}{\min(m_1,m_2)}},
	\]
	for a suitable constant $c$.

 Theorem \ref{thm:Depth2Grid} can be generalized to a constant $d$ dimensional grid. Specifically, the depth of any CNOT circuit can be optimized to $O\pbra{\frac{n^2\pbra{m_1 + \cdots + m_d}}{m_1\cdots m_d}}$ under $m_1\times \cdots \times m_d$ grid, where $d$ is a constant and $3n \leq m_1\cdots m_d\leq n^2$.
Let $m$ be the third largest value of $\cbra{m_1,\cdots, m_d}$. 

	\begin{figure*}[htbp]
	\begin{center}
\includegraphics[width=0.8\textwidth]{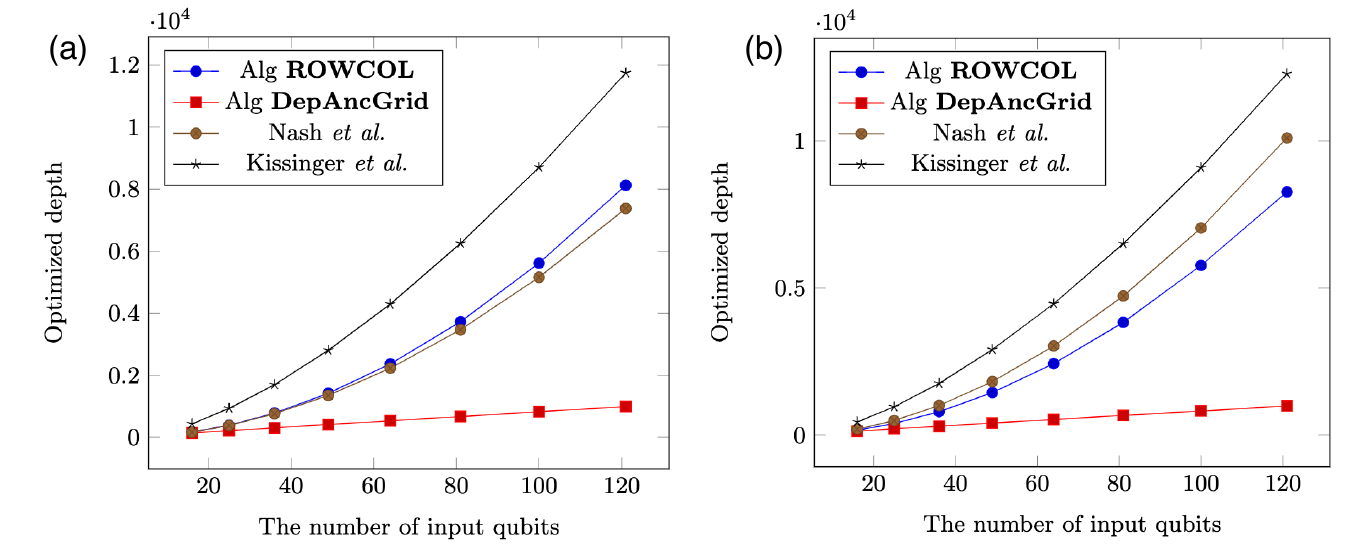}
\caption{{The comparison of optimized depth of these algorithms in the grid graph. We compare the performance of different synthesis algorithms in $n\times n$ grid. We use two different numbering of the grid: (a) numbering row by row from the first row to the last row, and (b) numbering from the inside to the outside in the form of a grid.}}
		\label{fig:exp_depth}
		\end{center}
	\end{figure*}
	
{\subsection{Numerical results}\label{subsec:DepthExperiment}
}

   In this subsection, we give the numerical tests for Algorithm \textbf{DepAncGrid}. We compare the optimized depth of Algorithm \textbf{DepAncGrid} with all of the existing size optimization algorithms on the grid graph mentioned previously. To show the performance of these algorithms, different sizes of $n\times n$ dimensional grids ranging from $4$ to $11$ are selected in this experiment. For one grid, we first randomly sample 200 different CNOT circuits, and then run all these algorithms under this condition. The method to sample a random CNOT circuit here is: (1). Randomly sample an $n\times n$ 0-1 matrix by randomly selecting a ``0'' or ``1'' in each position; (2). Determine whether the matrix sampled in (1) is invertible, if not, return to (1), otherwise accept the matrix as a random CNOT circuit.

   The comparison results are depicted in Fig. \ref{fig:exp_depth}, including the optimized depth of Algorithms \textbf{ROWCOL}, \textbf{DepAncGrid}, and algorithms in Refs. \cite{kissinger2019cnot,nash2019quantum} on grid graph.  In particular, here Algorithm \textbf{DepAncGrid} needs $n^2$ qubits and other algorithms only need $n$ qubits.
The $y$ axis shows the average depth of the optimized circuit. In consideration of reducing the impact of the different Hamiltonian paths chosen in Ref. \cite{kissinger2019cnot}, we choose two different Hamiltonian paths to synthesize the same CNOT circuit. 
    The comparison results show Algorithm \textbf{DepAncGrid} has a significant improvement for the optimized depth as the number of qubits increases.
    Theoretically, the depth of the CNOT circuit generated by Algorithm \textbf{DepAncGrid} equals $O(n)$, while $O(n^2)$ for other algorithms. 

    This experimental result shows that the depth of the CNOT circuit can be greatly reduced for the ancillas-free quantum system.

{\section{Experimental result on IBMQ}
\label{sec:experimentIBMQ}
}

In this section, we test the performance of our optimized CNOT circuits on IBM devices.
We implement a staircase CNOT circuit and an Add CNOT circuit, which have wide applications in error correction~\cite{nielsen2002quantum}, variational algorithms~\cite{nielsen2002quantum,cerezo2020variational} and quantum chemistry~\cite{hastings2014improving,sugisaki2019open}. 

We leverage IBMQ devices (ibmq\_athens and ibmq\_5\_yorktown) as the {limited connectivity architecture}~\cite{IBMQ2021}, as shown in Fig. S1 of the Supplementary material. In Fig. \ref{fig:example} (a), we give the staircase circuit without considering the {limited connectivity architecture}, with input $\ket{\phi} = \frac{\ket{0} + \ket{1}}{\sqrt{2}} \ket{0} \frac{\ket{0} + \ket{1}}{\sqrt{2}} \ket{0} \ket{0}$. We perform the circuit on IBMQ with the mapping the CNOT circuitprovided by the ROWCOL algorithm in Fig. \ref{fig:example} (b). There is a layer of $H$ gates before the CNOT circuit in Fig. \ref{fig:example} (b) to generate the input state $\ket{\phi}$ from initial state $\ket{0}^{\otimes 5}$ of the IBMQ device.

\begin{figure}[htbp]
        \centering
        \begin{small}
                \begin{tabular}[b]{c}
              \Qcircuit @C=0.6em @R=0.6em {
       \lstick{ q_0} &\ctrl{1} & \qw & \qw & \qw &
       \qw & \qw & \ctrl{1} & 
       \qw\\
       \lstick{ q_4} & \targ & \ctrl{1} & \qw & \qw & 
        \qw & \ctrl{1} & \targ &
       \qw\\
       \lstick{ q_1} & \qw & \targ & \ctrl{1} & \qw & 
        \ctrl{1} & \targ &\qw &
       \qw\\
       \lstick{ q_3} & \qw & \qw & \targ & \ctrl{1} & 
        \targ & \qw & \qw &
       \qw\\
       \lstick{ q_2} & \qw& \qw & \qw & \targ & 
       \qw & \qw & \qw&
       \qw
       } 
            \\
           \small (a)
        \end{tabular}
        \end{small}
        \qquad
        \begin{small}
                \begin{tabular}[b]{c}
              \Qcircuit @C=0.6em @R=0.6em {
       \lstick{ q_0} & \gate{H}      \barrier{4}&\ctrl{1} & \qw & \qw & \qw & \ctrl{1} & 
       \qw\\
       \lstick{ q_1} & \qw & \targ & \ctrl{1} & \qw  &  \ctrl{1} & \targ &
       \qw\\
       \lstick{ q_2} & \gate{H} & \qw & \targ & \targ  & \targ &\qw &
       \qw\\
       \lstick{ q_3} & \qw & \targ & \qw & \ctrl{-1} & \targ & 
        \qw & \qw  \\
       \lstick{ q_4} & \qw & \ctrl{-1}& \qw & \qw & \ctrl{-1} & 
       \qw & \qw 
       } 
            \\
            \small (b)
        \end{tabular}
        \end{small}
        \caption{{Mapping the staircase CNOT circuit to the {limited connectivity} 1D grid graph. (a) Staircase CNOT circuit with input state $\frac{\ket{0} + \ket{1}}{\sqrt{2}}\ket{0} \frac{\ket{0} + \ket{1}}{\sqrt{2}}\ket{0}\ket 0$. (b) A layer of $H$ gates, followed by a block of CNOT circuit, which is equivalent to the CNOT circuit of (a), and can be performed on the ibmq\_athens device (a 1D grid). The input state in (b) equals $\ket{0}^{\otimes 5}$.
        }}
        \label{fig:example}
    \end{figure}

We compared the measurement results of ROWCOL algorithm and IBM optimization algorithm on ibmq\_athens quantum device, and plot the classical simulation result (ideal quantum circuit without any error) as a comparison, as shown in Fig. \ref{fig:CNOTIBM}. We performed 8,000 measurements on each circuit independently. The horizontal axis shows the results measured under the computational basis, where $j$ represents computational basis $j_0j_1\ldots j_4$. The vertical axis shows the frequency of each outcome after 8,000 measurements. The ideal output state after performing the CNOT circuit $\Ccal$ in Fig. \ref{fig:example} (a) is $\ket{\psi} = \Ccal \ket{\phi} = \frac{\ket{0} + \ket{4} + \ket{16} + \ket{20}}{2}$. Hence the expected frequency of the result after 8,000 repeated measurements is $\cbra{\ket{0}:4000;\ket{4}:4000;\ket{16}:4000;\ket{20}:4000}$.
Fig. \ref{fig:example} shows that the simulation results are consistent with the expected frequency. It can also be seen from Fig. \ref{fig:CNOTIBM} that the circuit optimized by ROWCOL algorithm has a strong robustness to errors. We can extract the correct measurement result by setting a small threshold $y = 1000$ and selecting the outcomes whose frequency is greater than that threshold. As a comparison, there are tremendous errors in the circuit measurement results obtained directly by IBM's mapping method.

\begin{figure*}[htbp]
    \centering
    \includegraphics[width = 1.0\textwidth]{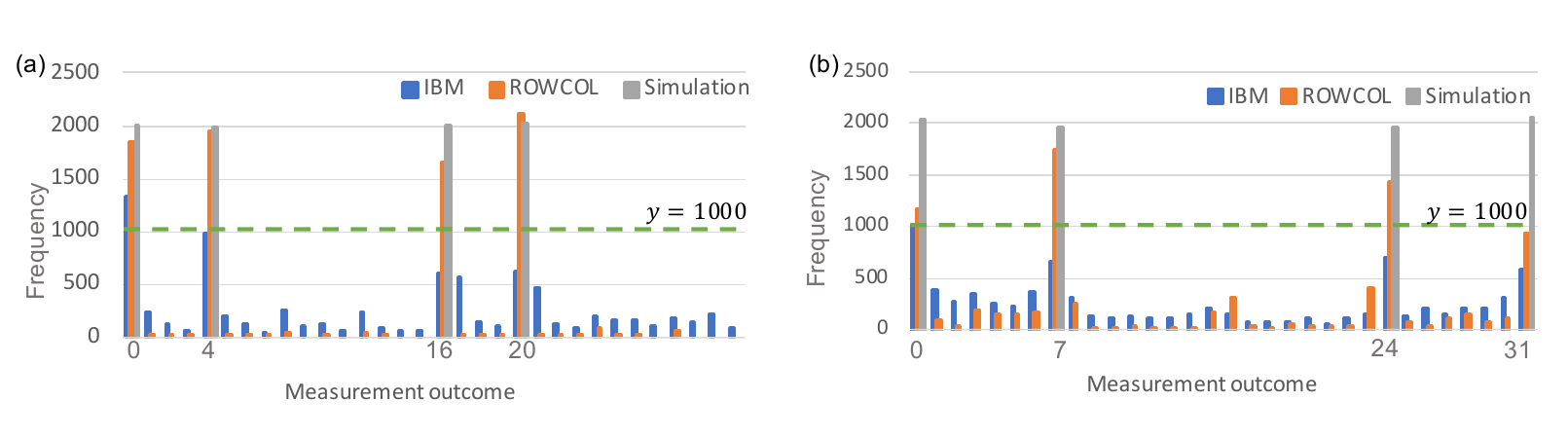}
    \caption{{The comparison of the ROWCOL algorithm, IBM optimization algorithm on IBM devices, and the classical simulation result.
    (a) For the CNOT circuit in Fig. \ref{fig:example} (a) under ibmq\_athens device. (b) For the CNOT circuit in Fig. \ref{fig:example_Add} (a) under ibmq\_5\_yorktown device.}}
    \label{fig:CNOTIBM}
\end{figure*}

The Add circuit, as shown in Fig. \ref{fig:example_Add} (a) and performed on the ibmq\_5\_yorktown device, has a similar performance, as shown in Fig. \ref{fig:CNOTIBM} (b).

\begin{figure}[htbp]
        \centering
        \begin{small}
                \begin{tabular}[b]{c}
              \Qcircuit @C=0.6em @R=0.6em {
       \lstick{ q_0} 
       &\ctrl{4} & \qw & \qw & \qw &
       \ctrl{3} & \qw & \qw & 
       \ctrl{2} & \qw &
       \ctrl{1} &\qw\\
       \lstick{ q_1} 
       & \qw & \ctrl{3} & \qw & \qw & 
        \qw & \ctrl{2} & \qw &
       \qw & \ctrl{1} &
       \targ & \qw \\
       \lstick{ q_2} 
       & \qw & \qw & \ctrl{2} & \qw & 
        \qw & \qw &\ctrl{1} &
       \targ & \targ &
       \qw & \qw\\
       \lstick{ q_3} 
       & \qw & \qw & \qw & \ctrl{1} & 
        \targ & \targ & \targ &
       \qw & \qw & 
       \qw & \qw\\
       \lstick{ q_4} &
       \targ& \targ& \targ & \targ &
       \qw & \qw & \qw&
       \qw & \qw & 
       \qw & \qw
       } 
            \\
           \small (a)
        \end{tabular}
        \end{small}
        \qquad
        \begin{small}
                \begin{tabular}[b]{c}
              \Qcircuit @C=0.6em @R=0.6em {
       \lstick{ q_0} & \gate{H}      \barrier{4}&\qw & \ctrl{2} & \qw & \ctrl{1} & \qw\\
       \lstick{ q_1} & \qw & \ctrl{1}& \qw & \qw  &  \targ & \qw\\
       \lstick{ q_2} & \gate{H} 
       & \targ & \targ &\ctrl{2} & \ctrl{1} &\qw\\
       \lstick{ q_3} & \qw &
       \ctrl{1} & \qw & \qw & \targ & 
        \qw \\
       \lstick{ q_4} & \qw & 
       \targ & \qw & \targ & \qw & 
       \qw
       } 
            \\
            \small (b)
        \end{tabular}
        \end{small}
        \caption{{Mapping the CNOT Add circuit to ibmq\_5\_yorktown device. (a) CNOT Add circuit with input state $\frac{\ket{0} + \ket{1}}{\sqrt{2}}\ket{0} \frac{\ket{0} + \ket{1}}{\sqrt{2}}\ket{0}\ket 0$. (b) A layer of $H$ gates, followed by a block of CNOT circuit, which is equivalent to the CNOT circuit of (a), and can be performed on the ibmq\_5\_yorktown device. The input state in (b) equals $\ket{0}^{\otimes 5}$.
        }}
        \label{fig:example_Add}
    \end{figure}

\section{Discussion}
\label{sec:discuss}

Optimization of the size/depth of the quantum circuit with {limited connectivity architecture} is one of the main challenges in near-term quantum computing~\cite{preskill2018quantum,paler2014mapping,paler2016synthesis}. In this paper, we propose two size/depth optimization algorithms on the {limited connectivity architecture}. 
The experimental results show our algorithms have better performances compared to the existing optimization algorithms.

{Specifically, for the connected graph which has a minimum degree $\delta$, any $n$-qubit CNOT circuits can be optimized to $O\pbra{\frac{n^2}{\log \delta}}$ size on a such graph. We also prove the order is tight for a regular graph.
Algorithm \textbf{SBE} further indicates the size of any $n$-qubit CNOT circuits can be optimized to $O\pbra{\frac{n^2}{\log (n/k)}}=O\pbra{\frac{n^2}{\log \delta}}$
on the {limited connectivity architecture} where the average degree of any $k$ vertex set is greater than $n/k$ for the {limited connectivity architecture}. }

{For the {limited connectivity architecture} which has a constant degree, we can optimize any CNOT
circuits to $2n^2$ size on any connected graph, the order is tight when the {limited connectivity architecture} is a simple path.
We also give an algorithm that takes more features of the graph into account and gives a better upper bound for the specific class of graphs.}

{
With these two size optimization algorithms, we
see that the optimized size of a CNOT circuit is dominated by the degrees of the limited connectivity structure. Here we give the intuition by ``implementation of a single CNOT gate''.
For the structure where the summation of any $k$ degrees from $k$ different vertices is larger than $n$, the number of required CNOT gates to construct a CNOT gate CNOT$_{ij}$ scales as $10k$ with Lemma \ref{lem:MinSteinerSize}. It implies that 
for the limited connectivity structure with minimum degree $\delta$, the number of required CNOT gates to construct a CNOT gate CNOT$_{ij}$ scales as $\frac{10n}{\delta}$
in the worst case. This phenomenon explains why the smaller optimized size is corresponding to stronger connectivity.
}

{Since the current quantum superconducting devices~\cite{IBMQ2021, Arute2019,boixo2018characterizing} usually have a low degree for the {limited connectivity architecture}, and the coherence time is very limited, we also consider a specifically {limited connectivity architecture} --- the $m_1 \times m_2$ grid. We show that any $n$-qubit CNOT circuits can be optimized to $O\pbra{\frac{n^2\pbra{m_1 + m_2}}{m_1m_2}}$ depth on this grid, where $3n \leq m_1,m_2\leq n^2$. This optimized result can be easily generalized to any constant $d$ dimensional grid.
The dimensions for a corresponding grid of current quantum superconducting devices~\cite{IBMQ2021,Arute2019,boixo2018characterizing} are $d\in \{1,2\}$. }

{
To conclude the suitable scenarios of the proposed algorithms, we list the options for the algorithms with the degrees of the associated limited connectivity structures as follows:
\begin{itemize}
    \item [(a)] The SBE algorithm has a better performance for the quantum computer with a comparatively large number of the minimum degree of the limited connectivity structure. This algorithm has a better performance compared to other existing algorithms for the cases where the minimum degree $\delta$ of the structure is larger than the constant, i.e., $\delta > O(1)$ and the number of input qubits are comparatively large (larger than hundreds of qubits).
    \item[(b)] The ROWCOL algorithm has a better performance for the near-term quantum device with dozens of qubits or the constant degree of the limited connectivity structure, such as 1D or 2D grids.
    \item [(c)] The DepAncGrid algorithm is suitable for 1D and 2D limited connectivity structures.
\end{itemize}}

{Note that the 65-qubit quantum superconducting device proposed by IBMQ~\cite{IBMQ2021} is not exactly a grid, nevertheless, it is a sub-graph of a grid. We can still perform Algorithm \textbf{DepAncGrid} on the expanded grid by leverage of SWAP gates to implement the CNOT gate for two vertices that are not connected in the sub-graph. Another avenue is to construct a new virtual 2D grid, and convert the single CNOT gate in the virtual grid to a series of CNOT gates in the real devices. We can ensure the additional cost of CNOT gates is bounded to a constant factor of the original one due to its sub-grid structure.}

{We list two open problems for the optimization of CNOT circuits on the {limited connectivity architecture}. 
(1) Is there any improved size optimization algorithm for some more specific structures under the {limited connectivity architecture}?
(2) If there are no ancillas, can we give some better results for the depth optimization of CNOT circuits on the $m_1 \times \cdots \times m_d$ grid for constant $d$?
}

{Similar to the CNOT gate, the Toffoli gate is also a classical reversible gate.
Nevertheless, the optimization of Toffoli gates is much more complex than CNOT circuits. An explicit reason is that it is not a linear map, so we can not construct the connection for the gate operations and the linear operations on a smaller size matrix, i.g., $\cbra{0,1}^{n\times n}$ boolean matrix. The intrinsic reason is that Toffoli circuits are computationally universal, and it generates all possible reversible transformations $f: \cbra{0,1}^n \rightarrow \cbra{0,1}^n$ if ancillas are allowed to be used~\cite{aaronson2017classification}. Since Toffoli is also a classical reversible gate, we can optimize specific Toffoli circuits instructed by the specific boolean function it implemented ---similar to the process with CNOT circuits. We give an example in the supplementary material for the implementation of the staircase of Toffoli gates under the full connectivity structure. We also leave an interesting open question of whether it can be applied to more general Toffoli circuits. 
}

\section*{Acknowledgement}
The authors thank Jiaqing Jiang for helpful discussions.
This work was supported in part by the National Natural Science Foundation of China Grants No. 61832003,  62272441, 61872334, 61801459, 12147133, the Strategic Priority Research Program of Chinese Academy of Sciences Grant No. XDB28000000, and Zhejiang Lab's International
Talent Fund for Young Professionals. 

\begin{appendix}
\section{Proof of Lemma 1}\label{app:ProfMSS}    

    \begin{proof}[Proof of Lemma 1]
   Denote $T_k(G) = \max_{|S|=k,S\subseteq V}T_G(S)$, where $T_G(S)$ is the size (the number of vertices) of the minimum Steiner tree for vertex set $S$ in $G$.
   
    Let $S=\cbra{u_1, \cdots, u_k} \subseteq V$ and $A$ be an empty set.
    In a connected graph $G(V,E)$, let the distance between two vertices be the number of edges in the shortest path that connects them. 
    Let $d(i,j)$ denote the distance between vertex $i$ and $j$ in Graph $G$, $d(A, v)$ denote the minimum distance between vertex $v$ and set $A$.
    
    Firstly put $u_1$ into set $A$, and then put all of $v\in V$ such that $d(A, v)= 3$ into $A$. That is,
    \begin{itemize}
        \item $d(v_i,v_j) \geq 3$ if $v_i, v_j \in A$.
        \item $d(A, v_i) \leq 2$ if $v_i \not\in A$.
    \end{itemize}
    Let $A'=\{a_1,\cdots, a_k\}$ be a set such that the element $a_i\in A'$ is a vertex in set $A$ and closest to $u_i$ in $S$.
    By the construction of $A$, we have
    $T_{G}(\cbra{a_1, \cdots, a_k})\leq 3(|A| - 1) + 1\leq 3|A|$.

    By the definition of $G$, for any $k$ vertex in $G$, the summation of their degrees is greater than $n$. On the other hand, since there are no common neighbors for any two vertex in set $A$, $\sum_{v\in A}d(v) + |A|\leq n$, thus we have $|A|<k$.
    
    Therefore,
   \[
    \begin{aligned}
    T_{G}(S)&\leq T_{G}(\cbra{a_1,\cdots, a_k}) + \sum_{j = 1}^{k} T_{G}(a_j, u_j)\\
    &\leq 3|A| + 2(k-1) \\
    &\leq 5k
    \end{aligned}
  \]
    \end{proof}

\section{Proof of Lemma 4}\label{app:profCanonical}
\begin{proof}
	For any $n$-qubit CNOT circuits with $k$ gates $R_1,R_2$, $\dots, R_k$, we denote their canonical forms by $\{G_1, G_2,\dots, G_s\}$, in which the lengths of the blocks are
	$ r_1,r_2,\dots,r_s$ respectively. We first consider the number of different partitions, \emph{i.e.}, different selections of $r_1,r_2,\dots, r_s$. It is easy to see the number is $2^{k-1}$ for any combination from set $[k-1]$, which is a partition of set $[k]$.
	
	Next, we derive the upper bound of the number of distinct canonical forms, given the specific partitioning $r_1,r_2,\dots, r_s$.
	
	For block 1, the index set of each matrix is required to be disjoint. Thus, there are at most $\dbinom{n}{r_1}$ different combinations of $r_1$ indices and there are at most $\Delta$ choices for another index of a matrix since the CNOT gates can only act on the nearest neighbor qubits and the maximum degree of the graph is $\Delta$. All this leaves for block 1  at most $(2\Delta)^{r_1} \dbinom{n}{r_1}$ possible combinations, where the factor 2 is due to the different order of the indices.

	Subsequently, for block 2, each index set of the matrix has at least one element intersecting with that of block 1, so we need to choose the $r_2$ index from the index set of block 1. The  number of possible combinations is $\dbinom{2r_1}{r_2}$. Similar to block 1, there are at most $(2\Delta)^{r_2} \dbinom{2r_1}{r_2}$ possible combinations.
	
	For the same reason, block $i$ has at most $(2\Delta)^{r_i}\dbinom{2r_{i-1}}{r_{i}}$ possible combinations. In all, the number of distinct canonical forms is at most
	\begin{equation}
	2^k \Delta^{k} \dbinom{n}{r_1}\dbinom{2r_1}{r_2}\dots \dbinom{2r_{s-1}}{r_s}.
	\end{equation}
	Since $\dbinom{n}{r_1} < n^{r_1}/r_1! < 2^{n\log n} / r_1! $, and $\dbinom{2r_i}{r_{i+1}} < 2^{r_{i+1}} r_i^{r_{i+1}} / r_{i+1}! $, we can relax the upper bound to
	\begin{equation}
	\frac{  4^k \Delta^{k} 2^{n\log n} r_1^{r_2} r_2^{r_3} \cdots r_{s-1}^{r_{s}}}{r_1! r_2! \cdots r_s!}.
	\end{equation}
	
From the Stirling formula, the following inequality holds
	\begin{equation}
	r_1!r_2!\cdots r_s! \ge \left(\frac{r_1}{e} \right)^{r_1}\left(\frac{r_2}{e} \right)^{r_2}\cdots \left(\frac{r_s}{e} \right)^{r_s}.
	\label{eq1}
	\end{equation}
And then applying the rearrangement inequality to obtain 
	\begin{equation}
	    r_1^{r_1} r_2^{r_2} \cdots r_{s}^{r_{s}}	\ge r_1^{r_2} r_2^{r_3} \cdots r_{s-1}^{r_{s}} r_{s}^{r_{1}} \ge r_1^{r_2} r_2^{r_3} \cdots r_{s-1}^{r_{s}}, 
	    \label{eq2}
	\end{equation}
The last inequality holds for $r_{s}\ge 1 $ and $r_{1} \ge 1$. Combining Eq. (\ref{eq1}) and (\ref{eq2}), we have the following inequality, 
	\begin{equation}
	    r_1!r_2!\cdots r_s! \ge  e^{-k} r_1^{r_2} r_2^{r_3} \cdots r_{s-1}^{r_{s}}.
	\end{equation}
	Therefore, we can obtain the desired upper bound by multiplying the number of different partitioning ways.
	\begin{equation}
	 8^k e^{k} \Delta^{k} 2^{n\log n}.
	\end{equation}
\end{proof}

\section{Algorithms for copy operation and add operation\label{app:CopyAdd}}

\begin{proof}[Proof of Lemma 5]
Each qubit in the hypergrid has a $d$ dimensional grid coordinate $(p_1,p_2,\dots,p_d)$, $p_i\in [m^{1/d}]$.  We may assume $q_1$ in $(1,1,\dots,1)$ without loss of generality, otherwise $q_1$ can be moved to $(1,1,\dots,1)$ by at most $dm^{1/d}$ SWAP gates, where every SWAP gate can be implemented by $3$ CNOT gates. We can copy $q_1$ to all other vertices by following operations, we use '$\leftarrow$' to denote CNOT from qubit at right coordinate to left coordinate.
\begin{enumerate}
\item for $k$ from $2$ to $m_1$, $(k,1,\dots,1)\leftarrow (k-1,1,\dots,1)$,
\item for every $p_1\in [m_1]$, for $k$ from $2$ to $m_2$, $(p_1,k,1,\dots,1)\leftarrow (p_1,k-1,1,\dots,1)$,
\item for every $p_1\in[m_1],p_2\in [m_2]$, for $k$ from $2$ to $m_3$, $(p_1,p_2,k,1,\dots,1)\leftarrow (p_1,p_2,k-1,1,\dots,1)$,
\item[] $\dots$
\item[\textbf{d.}] for every $p_i\in[m_i]$,$1\le i\le d-1$, for $k$ from $2$ to $m_d$, $(p_1,\dots, p_{d-1}, k)\leftarrow (p_1, \dots, p_{d-1},  k-1)$,
\end{enumerate}
we copy $q_1$ to all other vertices by the above operations, operations in the $i$-th item can be parallelized to $m_i-1$ depth, hence the  total depth is $O(\sum_{i=1}^dm_i)$.
\end{proof}

\begin{proof}[Proof of Lemma 6]
By the same argument in proof of Lemma 5, we may assume the $n$-th vertex in $(1,1,\dots,1)$ without loss of generality. We can implement the addition operation as follows.
\begin{enumerate}
\item For $k$ from $2$ to $m_1$, if $ (k,1,\dots,1)\notin S$, then $(k-1,1,\dots,1)\leftarrow (k,1,\dots,1)$,
\item For every $p_1\in [m_1]$, for $k$ from $2$ to $m_2$, if $ (p_1,k,1,\dots,1)\notin S$, then $(p_1,k-1,1,\dots,1)\leftarrow (p_1,k,1,\dots,1)$,
\item For every $p_1\in[m_1],p_2\in[m_2]$, for $k$ from $2$ to $m_3$, if $(p_1,p_2,k,1,\dots,1)\notin S$, then $(p_1,p_2,k-1,1,\dots,1)\leftarrow (p_1,p_2,k,1,\dots,1)$,
\item[] $\dots$
\item[\textbf{d.}] For every $p_i\in[m_i]$, $1\le i\le d-1$, for $k$ from $2$ to $m_d$, if $(p_1,p_2,\dots,p_{d-1},k)\notin S$, then $(p_1,p_2,\dots,p_{d-1},k-1)\leftarrow (p_1,p_2,\dots,p_{d-1},k)$,
\item[\textbf{d+1.}] For every $p_i\in[m_i]$, $1\le i\le d-1$, for $k$ from $m_d$ to $2$, $(p_1,p_2,\dots,p_{d-1},k-1)\leftarrow (p_1,p_2,\dots,p_{d-1},k)$,
\item[] $\dots$
\item[\textbf{2d-2.}] For every $p_1\in[m_1],p_2\in [m_2]$, for $k$ from $m_3$ to $2$, $(p_1,p_2,k-1,1,\dots,1)\leftarrow (p_1,p_2,k,1,\dots,1)$,
\item[\textbf{2d-1.}] For every $p_1\in[m_1]$, for $k$ from $m_2$ to $2$, $(p_1,k-1,1,\dots,1)\leftarrow (p_1,k,1,\dots,1)$,
\item[\textbf{2d.}] For $k$ from $m_1$ to $2$, $(k-1,1,\dots,1)\leftarrow (k,1,\dots,1)$,
\end{enumerate}
we get the sum of all vertices, except the vertices not in $S$ added twice. The arithmetic operations are over $\Fbb_2$, so $y=\sum_{i\in S} x_i$. We can recover $x_1,x_2,\dots,x_{n-1}$ by reversing above operations which are not related to $(1,1,\dots,1)$. The operations in the $i$-th item and $(2d+1-i)$-th item can be parallelized to $m_i-1$ depth, where $1\le i\le d$, hence the total depth is $O(\sum_{i=1}^dm_i)$.
\end{proof}

\section{The constrained graphs for the IBM devices for our experiments}\label{app:ibm_device}

 \begin{figure}[htbp]
    \centering
    \includegraphics[width = 0.5\textwidth]{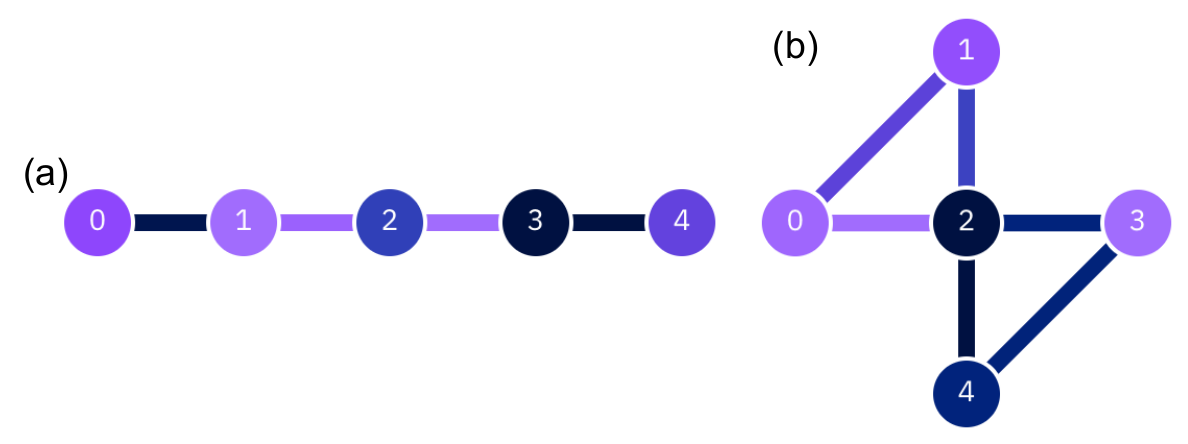}
    \caption{The topological constrained graph for (a) ibmq\_athens device, and (b) ibmq\_5\_yorktown device.}
    \label{fig:ibm_device}
\end{figure}

\section{Optimization of the staircase of Toffoli gates}
\label{sec:toffoli_staircase}

In this section, we give the optimization method for staircase of Toffoli gates.

\begin{figure}[ht]
    \centering
    \[\Qcircuit @C=1em @R=.7em {
    & \ctrl{2} & \qw & \qw & \qw & \qw \\
    & \ctrl{1} & \qw & \qw & \qw & \qw \\
     & \targ & \ctrl{2} & \qw & \qw & \qw \\
     & \qw & \ctrl{1} & \qw & \qw & \qw \\
     & \qw & \targ & \ctrl{2} & \qw & \qw \\
     & \qw & \qw & \ctrl{1} & \qw & \qw \\
     & \qw & \qw & \targ & ^{\cdots} \qw & \qw \\
   } \]
    \caption{Staircase of Toffoli gates.}
    \label{fig:StairToffoli}
\end{figure}

\begin{lemma}
The staircase of Toffoli gates (as in Fig. \ref{fig:StairToffoli}) can be paralleled to $O(\log n)$ depth with $O(n^3)$ ancillas under the full connectivity structure.
\end{lemma}

\begin{proof}
Let unitary $\Ccal$ denote the circuit of the staircase of Toffoli gates (Figure \ref{fig:StairToffoli}), and let output $y:= \Ccal x$, where $x\in \cbra{0,1}^n$ is the input string,  then we have
$$y_k =
\begin{cases}
\sum_{0\leq j \leq k/2} x_{2j} \prod_{j\leq i \leq k/2 - 1}x_{2i+1}, \text{ if } k\text{ mod } 2 = 0,\\
x_k, \text{ otherwise.}
\end{cases}$$
for any $0\leq k \leq n - 1$. It's easy to find that the frequency of $q_i$ in all of $y_k$ is
$$\# x_{2k - 1} = k (n - k),\quad \# x_{2k} = n - k. $$
where $0\leq k \leq \frac{n - 1}{2}$.

We can construct all of the required outputs by
\begin{itemize}
    \item [(2)] Computing all of the multiplications of the variables in every clause in parallel with all of the copied variables in the ancillas;
    \item [(3)] Computing all of the addition of all the clauses in parallel;
    \item[(4)] Restore all of the qubits except the generated output $y$.
    \item [(5)] Let $U_{f}$ be the circuit of Steps (1-4). Restore all of ancillas to zeros by performing $U_{f^{-1}}$ to qubits $(y,x,\textbf{0})$, where $y = f(x)$ is the associated boolean functions of the staircase of Toffoli gates, and $U_f$ is the circuit constructed by the sub-circuits from Steps (1-4). 
\end{itemize}

If we are allowed to use CNOT gates, then the copy operation can be implemented with CNOT circuits, otherwise, we can use a Toffoli gate with one control qubit initialized as one to implement the CNOT gate.
Hence the copy and multiplication operations
$x_0 x_1 \cdots x_n$ or $x_0 \oplus x_1 \oplus \cdots \oplus x_n$ 
both can be implemented in $O(\log n)$ and $O(n)$ ancillary qubits.
For each $y_k$ we require $O(n^2)$ ancillas, and Steps (1-4) require $O(n^3)$ ancillas in total.
Here the ${f^{-1}}$ can be implemented as the inverse of the circuit in Fig.~\ref{fig:StairToffoli}. $U_{f^{-1}}$ can be implemented with $O(n)$ size of Toffoli gates since the expression of $f^{-1}(x)$ 
is much easier.
\end{proof}

\end{appendix}

\end{document}